\documentclass[article]{IEEEtran}

\usepackage{xcolor,comment}
\usepackage{hyperref}
\usepackage{amsmath,amssymb,amsthm,bm}

\theoremstyle{remark}

\newcommand{\Utrue}{\mathbf {U}_{S_3}}
\newcommand{\Cplx}{\mathbb C}
\newcommand{\Prob}{\mathbb P}
\newcommand{\norm}[1]{\lVert #1 \rVert}
\newcommand{\herm}{H}
\newcommand{\E}{\mathbb{E}}

\newcommand{\Uhat}{\widehat{\mathbf{U}}_{\mathrm{s,mat}}}
\newcommand{\Umat}{\widehat{\mathbf U}_{\mathrm{s,mat}}}
\newcommand{\Uten}{\widehat{\mathbf{U}}_{\mathrm{s,ten}}}
\newcommand{\PiS}{\Utrue\Utrue^{\herm}}
\newcommand{\Epar}{\mathbf{E}_{\parallel}}

\newcommand{\sigmaL}{\sigma_{L}}
\newcommand{\lmax}{\lambda_{\max}}
\newcommand{\lmin}{\lambda_{\min}}
\DeclareMathOperator{\range}{range}

\newcommand{\Proj}{\mathbf{P}}

\newcommand{\Ksub}{\mathcal{K}}
\newcommand{\Ksubhat}{\widehat{\mathcal{K}}}
\newcommand{\PK}{\mathbf{P}_{\Ksub}}
\newcommand{\PKhat}{\mathbf {P}_{\Ksubhat}}

\newcommand{\Uorc}{\widetilde{\mathbf U}_{\mathrm{s,ten}}}

\newcommand{\gor}{g_{\mathrm{oracle}}}
\newcommand{\gorlower}{g_{\mathrm{oracle}}^{\mathrm{lower}}}
\newcommand{\gorupper}{g_{\mathrm{oracle}}^{\mathrm{upper}}}
\newcommand{\lem}{\ell_{\mathrm{emp}}}
\newcommand{\lemupper}{\ell_{\mathrm{emp}}^{\mathrm{upper}}}
\newcommand{\blower}{b_{\mathrm{lower}}}

\newcommand{\spec}[1]{\lVert #1\rVert_2}
\DeclareMathOperator{\rank}{rank}
\newcommand{\I}{\mathbf{I}}

\DeclareMathOperator{\tr}{tr}

\newcommand{\Uperp}{\mathbf{U}_{S_3}^{\perp}}
\newcommand{\UperpH}{\mathbf{U}_{S_3}^{\perp\herm}}
\newcommand{\Vhat}{\widehat{\mathbf{V}}_{\mathrm{s,mat}}}
\newcommand{\Xt}{\mathbf{X}_3}
\newcommand{\St}{\mathbf{S}_3}
\newcommand{\Qt}{\mathbf{Q}_3}

\newcommand{\vL}{\mathbf{v}_L}
\newcommand{\vLh}{\widehat{\mathbf{v}}_L}
\newcommand{\uLh}{\widehat{\mathbf{u}}_L}
\newcommand{\Sig}{\boldsymbol{\Sigma}}
\newcommand{\mMK}{\min(M,K)}
\newcommand{\twonorm}[1]{\left\lVert #1\right\rVert_2}
\newcommand{\Shat}{\widehat{\boldsymbol{\Sigma}}_{\mathrm{s,mat}}}
\newcommand{\fnorm}[1]{\left\lVert #1\right\rVert_F}

\usepackage{graphicx}
\usepackage{amsmath}
\usepackage{bm}
\usepackage{footnote}
\usepackage{amssymb}
\usepackage{subfigure}
\usepackage{cite}
\usepackage{algorithm}
\usepackage{algorithm,algcompatible,amsmath}
\usepackage{algpseudocode}
\usepackage{xcolor}
\usepackage{booktabs}
\usepackage{cuted}
\usepackage{extarrows}
\usepackage{stmaryrd}
\newtheorem{theorem}{Theorem}[section]
\newtheorem{corollary}[theorem]{Corollary}
\newtheorem{lemma}[theorem]{Lemma}
\newtheorem{proposition}[theorem]{Proposition}

\title{
Tensor-Based Joint Pitch and DOA Estimation
}

\author{Yifan Wu and Michael B. Wakin
\thanks{Yifan Wu is with the Department of Electrical and Computer Engineering, California State Polytechnic University, Pomona, CA, 91768, USA (e-mail: yiw062@ucsd.edu, yifanwu@cpp.edu).
       }%
\thanks{Michael B. Wakin is with the Department of Electrical Engineering, Colorado School of Mines, Golden, CO, 80401, USA (e-mail:mwakin@mines.edu).}%
}

\begin{document}

\maketitle
\thispagestyle{empty}

\begin{abstract}

We propose a tensor-based subspace method for the joint estimation of pitch and direction of arrival (DOA) of multiple harmonic sources. Unlike conventional matrix-based approaches that unfold the spatio-temporal data into a single matrix, the proposed method preserves the multidimensional Hankel tensor structure and exploits a Kronecker product constraint satisfied by the true signal subspace. The matrix-based subspace estimate is refined by projecting it onto an empirically estimated Kronecker-structured signal cage obtained from mode-unfolded tensor subspaces. Beyond the estimator itself, we provide a non-asymptotic theory explaining when this tensor refinement improves the matrix baseline. We show that the oracle refinement using the true Kronecker projector is non-worsening, decompose the overall tensor gain into oracle gain and empirical loss, and derive deterministic and high-probability sufficient conditions for positive overall tensor gain under complex Gaussian noise. The analysis reveals a sweet-spot behavior: the largest certified gain occurs when the matrix estimate is neither nearly perfect nor too inaccurate. Simulations over closely spaced, harmonically related, and coherent-source scenarios show that the proposed method improves subspace accuracy and downstream pitch/DOA estimation, and also outperforms representative matrix and tensor-train-based baselines.




\end{abstract}

\section{INTRODUCTION}

The estimation of the fundamental frequency, or pitch, of a signal is a cornerstone of speech and audio processing~\cite{kim2025neural,haeb2025microphone,liu2025unsupervised,lehmann2025comprehensive}. Because real-world scenarios often involve multiple concurrent harmonic sources \cite{christensen2009multi}, research has increasingly shifted toward joint pitch and direction-of-arrival (DOA) estimation using multi-channel sensor arrays \cite{christensen2009multi, wu2014joint, karimian2016computationally, jensen2013nonlinear, zhang2012joint, zhou2019parametric}. Early joint estimation methods were often constrained by narrowband signal assumptions or high computational costs associated with multi-dimensional grid searches. Progress in this field has seen the development of more efficient subspace-based techniques such as Estimation of Signal Parameters via Rotational Invariance Techniques (ESPRIT)~\cite{wu2014joint}, which leverages temporal and spatial shift-invariance structures. More recently, the focus has shifted toward enhancing robustness in non-ideal conditions through filtering methods and minimum variance solutions that account for colored noise statistics.

Meanwhile, tensor-based techniques \cite{haardt2008higher, roemer2014analytical, cheng2014tensor, rakhimov2023analytical, xu2022doa, zheng2021coupled,  goulart2015tensor, wu2009joint, sidiropoulos2017tensor, cichocki2015tensor, miron2020tensor, tokcan2025tensor, auddy2025tensors} have significantly advanced multi-dimensional harmonics estimation by exploiting the inherent structure of signals received by sensor arrays. In \cite{haardt2008higher}, a higher-order SVD (HOSVD) based method was proposed. By estimating the signal subspace directly from the tensor instead of ``stacking'' multidimensional data into matrices, this approach preserves the inherent multi-dimensional structure from the start, leading to significantly more accurate subspace and parameter estimates. In \cite{roemer2014analytical}, a generic framework to assess the performance of tensor methods was established, and the ``tensor gain'' over matrix methods was empirically demonstrated. In \cite{cheng2014tensor}, a general framework (TeTraKron) was introduced that extends arbitrary matrix-based subspace tracking schemes to their tensor-based counterparts via a Kronecker-structured projection applied to the matrix-based subspace estimate, yielding improved accuracy in time-varying multidimensional harmonic retrieval.  A closed-form analytical performance assessment of 2-D Tensor ESPRIT was derived in terms of physical parameters in \cite{rakhimov2023analytical}, expressing the mean squared estimation error as a function of the SNR, array geometry, number of snapshots, and signal correlation to quantify the accuracy difference between the tensor and matrix versions of ESPRIT. In \cite{xu2022doa}, a novel tensor-based framework for high-resolution DOA estimation in transmit beamspace (TB) MIMO radar systems was proposed. In \cite{zheng2021coupled}, the authors proposed a coupled coarray tensor canonical polyadic (CP) decomposition-based method for DOA estimation that exploits the multidimensional structure of coprime L-shaped arrays. 
The authors in \cite{goulart2015tensor} developed algorithms for tensor CP decomposition with structured factor matrices and analyzed their performance, demonstrating how incorporating structural constraints improves identifiability and estimation accuracy. A joint time-delay and frequency estimation approach based on parallel factor analysis was proposed in \cite{wu2009joint}. We refer readers to \cite{sidiropoulos2017tensor, cichocki2015tensor, miron2020tensor, tokcan2025tensor, auddy2025tensors} for a comprehensive review of the progress of tensor methods in signal processing problems. 

Despite this widespread success, the advantage of tensor-based subspace estimation is almost universally invoked as a qualitative principle---the ``tensor gain'' attributed to the algebraic structure of the multidimensional
data---rather than quantified or guaranteed. The phenomenon was recognized early: the HOSVD-based estimate of \cite{haardt2008higher} was shown to improve
subspace accuracy, and \cite{thakre2010tensor} demonstrated that a tensor-based estimate is a strictly better estimate of the signal subspace than its matrix-based counterpart in the presence of noise. The most complete relevant theoretical treatment to date is the analytical performance-assessment framework of \cite{roemer2014analytical}, which provides explicit first-order expansions of the subspace estimation error for both matrix- and tensor-based ESPRIT-type algorithms and has become the standard tool for predicting their accuracy. Its reach, however, is fundamentally \emph{asymptotic}---the expansions derive from a first-order perturbation in the noise and are exact only in the high-SNR or large-sample regime---and \emph{descriptive}, in that it predicts each estimator's error and, in the noiseless limit, identifies when the two subspace estimates coincide, but does not certify that the gain is strictly positive at a finite operating point. Related efforts likewise remain asymptotic in nature, whether characterizing the maximum gain of spatial smoothing~\cite{steinwandt2017performance} or the perturbation of the HOSVD itself~\cite{zhang2018tensor}. In short, the existence of the tensor gain has been established and its asymptotic behavior analyzed, yet the conditions guaranteeing it for a fixed geometry, source configuration, and noise level have remained largely open.

In this work, we propose a tensor-based framework for joint pitch and DOA estimation. By modeling the received signals as a multi-dimensional tensor and utilizing the inherent Kronecker structure of the pitch and DOA harmonics, we propose a refined subspace estimation technique that utilizes the Kronecker product of projection matrices from mode-unfolded data. We demonstrate that this tensor-based refinement provides superior subspace estimation as well as pitch and DOA estimation accuracy compared to its matrix-based counterpart, particularly in noisy environments, thereby establishing a more robust foundation for subsequent pitch and DOA retrieval.

Beyond proposing the estimator, we also provide a theoretical analysis explaining when this refinement improves the matrix-based subspace estimate. We first show that the ideal, oracle version of the tensor refinement—formed using the true Kronecker projector—is never worse than the matrix-based subspace estimate whenever the matrix estimate has nonzero overlap with the true signal subspace. We then decompose the empirical tensor gain into two competing terms: an oracle gain, which measures the improvement obtained from the exact Kronecker constraint, and an empirical loss, which accounts for the error incurred when the Kronecker projector is estimated from noisy data. This decomposition yields a deterministic sufficient condition for positive tensor gain. 

We further convert this deterministic condition into a high-probability guarantee under additive complex Gaussian noise. The analysis combines Wedin-type subspace perturbation bounds, concentration inequalities for the structured Hankel noise unfoldings, and a bound on the leakage of the matrix-based subspace outside the Kronecker-constrained signal cage. 
This leads to an explicit non-asymptotic condition under which the tensor-refined subspace is guaranteed to be closer to the true signal subspace than the matrix-based estimate.

In summary, the main contributions of this paper are as follows: (1)~We formulate joint pitch and DOA estimation as a structured tensor subspace problem and identify a Kronecker product constraint satisfied by the true signal subspace. (2)~We propose a tensor-refined subspace estimator that projects the conventional matrix-based estimate onto an empirically estimated Kronecker-structured signal cage. (3)~We establish an oracle tensor-gain result showing that the ideal tensor refinement cannot degrade the matrix-based subspace estimate. (4)~We derive deterministic and high-probability sufficient conditions for positive empirical tensor gain, explicitly balancing oracle improvement against empirical Kronecker-projector estimation error. (5)~We demonstrate through simulations that the proposed method improves both subspace accuracy and downstream pitch/DOA estimation accuracy across several challenging acoustic scenarios. We also numerically demonstrate the validity of the theoretical analysis. To the best of our knowledge, this is the first work to provide explicit deterministic and high-probability sufficient conditions under which a Kronecker-structured tensor refinement is guaranteed to improve a given matrix-based subspace estimate.

The remainder of this paper is organized as follows. Section~\ref{sec:sigmod} introduces the signal model, the matrix and tensor data structures, and some useful decompositions of the unfolded tensors. Section~\ref{sec:methodology} first presents subspace-based estimation methods formulated using both matrix and tensor representations, and subsequently develops joint pitch and DOA estimation techniques based on ESPRIT. Section~\ref{sec:guarantee} presents a deterministic sufficient condition for positive tensor gain and then converts the deterministic condition into a non-asymptotic high probability guarantee under additive complex Gaussian noise. Section~\ref{sec:examples} provides numerical examples that demonstrate the performance advantages of the tensor-based method. Finally, Section~\ref{sec:conclusion} concludes. 

Throughout the paper, we use lowercase boldface for vectors ($\mathbf{a}$), uppercase boldface for matrices ($\mathbf{A}$), and calligraphic script for tensors
($\mathcal{X}$). Scalars are written in italics ($a$, $L$). $j = \sqrt{-1}$. For a matrix $\mathbf{A}$, we write $\mathbf{A}^{\top}$,
$\mathbf{A}^{H}$, $\mathbf{A}^{-1}$, and $\mathbf{A}^{+}$ for its transpose, conjugate transpose, inverse, and Moore--Penrose pseudoinverse, respectively; $\sigma_i(\mathbf{A})$ denotes the $i$-th largest singular value; 
and $\|\mathbf{A}\|_2$ and $\|\mathbf{A}\|_{F}$ denote the spectral and Frobenius norms, respectively. Finally, letting $\mathbf F$ and $\mathbf G$ be matrices of the same dimension, we measure the distance between the subspaces $\range(\mathbf F)$ and $\range(\mathbf G)$ as:
\[
d(\mathbf F,\mathbf G) := \spec{\Proj_{\range(\mathbf F)}-\Proj_{\range(\mathbf G)}},
\]
where $\Proj_{\range(\mathbf F)} = \mathbf F\mathbf F^{+}$ and $\Proj_{\range(\mathbf G)}  = \mathbf G\mathbf G^{+}$ denote orthogonal projection matrices onto $\range(\mathbf F)$ and $\range(\mathbf G)$, respectively. The subspace distance coincides with the sine of the largest principal angle between $\range(\mathbf F)$ and $\range(\mathbf G)$. Note that when $\mathbf F$ has orthonormal columns, $\Proj_{\range(\mathbf F)}=\mathbf F\mathbf F^{H}$, and when $\mathbf G$ has orthonormal columns, $\Proj_{\range(\mathbf G)}=\mathbf G\mathbf G^{H}$.

\section{Signal Model}
\label{sec:sigmod}

\subsection{Data}

\subsubsection{Setup}

We suppose that $R$ microphones are arranged in a uniform linear array (ULA) with array spacing $d$. The microphones synchronously collect $N$ uniform time samples, which we denote as $x_r(n)$ for microphone $r \in \{0,1,\dots,R-1\}$ and sample index $n \in \{0,1,\dots,N-1\}$.

We assume that $P$ sources are incident on the microphone array. Each source $p \in \{1,2,\dots,P\}$ is harmonic with a known harmonic number (model order) denoted by $L_p$. We let $\omega_p$ denote the unknown pitch of source $p$ and $\alpha_{p, l}$ denote the unknown complex amplitude of harmonic $l \in \{1,2,\dots,L_p\}$ for source $p$. Finally, we let $\theta_p$ denote the unknown DOA of source $p$. 

We can therefore model the received data as 
\[
x_r(n) = s_r(n) + q_r(n),
\]
where 
\begin{equation}
\label{s_r}
s_r(n) = \sum_{p = 1}^P \sum_{l = 1}^{L_p} \alpha_{p, l}e^{j(\omega_p ln + \phi_p lr)}
\end{equation}
and $q_r(n)$ is additive complex white Gaussian noise with variance $\sigma^2$. In~\eqref{s_r}, $\phi_p = \omega_p f_s \frac{d}{c} \sin(\theta_p)$ is the phase shift for source $p$ caused by the propagation time delay between array elements, where $f_s$ is the signal sampling frequency, and $c$ is the speed of propagation.

\subsubsection{Matrices}

For a positive integer $M$ that we refer to as the data matrix dimension, we can arrange the received data into a series of $R \times M$ matrices
\begin{equation*}
    \mathbf{X}(n) = 
    \begin{bmatrix}
x_0(n) & \dots & x_0(n + M - 1)\\
\vdots & \ddots & \vdots \\
x_{R-1}(n) & \dots & x_{R-1}(n + M - 1)
\end{bmatrix} \in \mathbb{C}^{R \times M}
\end{equation*}
for $n = 0, 1, ..., N-M$. Similarly, $s_r(n)$ and $q_r(n)$ may be arranged into $\mathbf{S}(n)$ (signal matrices) and $\mathbf{Q}(n)$ (noise matrices), respectively. We then have $\mathbf{X}(n) = \mathbf{S}(n) + \mathbf{Q}(n)$, $n = 0, 1, ..., N-M$. The signal matrices $\mathbf{S}(n)$ can be rewritten using \eqref{s_r} as 
\begin{equation*}
    \mathbf{S}(n) = \sum_{p = 1}^P \sum_{l = 1}^{L_p} \beta_{p, l}(n) \mathbf{a}_s(l\phi_p) \mathbf{a}_t^T(l\omega_p),
\end{equation*}
where $\beta_{p, l}(n) = \alpha_{p, l} e^{j l \omega_p n}$, and 
\[
\mathbf{a}_s(\phi_p) = [1 \quad e^{j\phi_p} \quad \cdots \quad e^{j(R -1 )\phi_p} ]^T
\]
and 
\[
\mathbf{a}_t(\omega_p) = [1 \quad e^{j\omega_p} \quad \cdots \quad e^{j(M-1)\omega_p}]^T
\]
are steering vectors.
 
\subsubsection{Tensors}

Stacking all of the received data matrices together, we can construct the  $R \times M \times (N - M + 1)$ tensor 
\[
\mathcal{X} = [\mathbf{X}(0) | \mathbf{X}(1) | \cdots | \mathbf{X}(N-M)] \in \mathbb{C}^{R \times M \times (N - M + 1)}.
\]
Similarly, $\mathbf{S}(n)$ and $\mathbf{Q}(n)$ may be arranged into $\mathcal{S}$ (signal) and $\mathcal{Q}$ (noise) tensors, respectively. We then have $\mathcal{X} = \mathcal{S} + \mathcal{Q}$.

The $R \times M \times (N - M + 1)$ signal tensor $\mathcal{S}$ can be unfolded into the mode-1, mode-2, and mode-3 unfolding matrices $\mathbf{S}_1 = [\mathcal{S}]_{(1)} \in \mathbb{C}^{R \times M(N-M+1)}$, $\mathbf{S}_2 = [\mathcal{S}]_{(2)} \in \mathbb{C}^{M \times R(N-M+1)}$, and $\mathbf{S}_3 = [\mathcal{S}]_{(3)} \in \mathbb{C}^{RM \times (N-M+1)}$, respectively.
Note that $\mathbf{S}_1  = [\mathbf{S}(0) ~ \mathbf{S}(1) ~ \cdots ~ \mathbf{S}(N-M)]$ and $\mathbf{S}_3  = [\mathrm{vec}(\mathbf{S}(0)) ~ \mathrm{vec}(\mathbf{S}(1)) ~ \cdots ~ \mathrm{vec}(\mathbf{S}(N-M))]$. Analogous mode-1, 2, and 3 unfoldings can be defined for $\mathcal{X}$ and~$\mathcal{Q}$.

\subsection{Useful decompositions}
\label{sec:decomps}

\subsubsection{Steering matrices}

Letting 
\[
\mathbf{a}(\omega_p, \phi_p) = \mathbf{a}_t(\omega_p) \otimes \mathbf{a}_s(\phi_p) \in \mathbb{C}^{RM \times 1}
\]
denote the Kronecker product of $\mathbf{a}_t(\omega_p)$ and $\mathbf{a}_s(\phi_p)$, we define the matrix 
\begin{align*}
\mathbf{A} &= [\mathbf{a}(\omega_1, \phi_1), \dots, \mathbf{a}(L_1\omega_1, L_1\phi_1), \mathbf{a}(\omega_2, \phi_2), \dots \\ &\qquad\qquad \dots, \mathbf{a}(L_P\omega_P, L_P\phi_P)]
\end{align*}
of size ${RM \times L}$, where 
\[
L := \sum_{p=1}^P L_p
\]
is the total number of harmonic components. Similarly, we can define the $L \times (N - M + 1)$ matrix
$\mathbf{Z} := [\mathbf{z}(0) \; \mathbf{z}(1) \; \cdots \; \mathbf{z}(N-M)]$, 
where $\mathbf{z}(n) = [\beta_{1,1}(n), \dots, \beta_{1,L_1}(n), \beta_{2,1}(n), \dots, \beta_{P, L_P}(n)]^T$. With $\mathbf{A}$ and $\mathbf{Z}$, we can then write
\begin{equation}
\label{S3A}
    \mathbf{S}_3 = \mathbf{A}\mathbf{Z}.
\end{equation}

Define the matrices
\begin{align*}
    \mathbf{A}_t &= [\mathbf{a}_t(\omega_1), \dots, \mathbf{a}_t(L_1\omega_1), \mathbf{a}_t(\omega_2), \dots, \mathbf{a}_t(L_P\omega_P)] \in \mathbb{C}^{M \times L}, \\
    \mathbf{A}_s &= [\mathbf{a}_s(\phi_1), \dots, \mathbf{a}_s(L_1\phi_1), \mathbf{a}_s(\phi_2), \dots \mathbf{a}_s(L_P\phi_P)] \in \mathbb{C}^{R \times L}.
\end{align*}
An important observation is that
\begin{equation}
    \mathbf{A} = \mathbf{A}_t * \mathbf{A}_s,
    \label{eq:AKRprod}
\end{equation}
where $*$ denotes the Khatri-Rao product. 

We can also write 
\begin{align}
\label{S1AsANDS2At}
\mathbf{S}_{1} &= \mathbf{A}_s (\mathbf{Z} \bullet \mathbf{A}_t^{\top}) ~\text{and}~ \mathbf{S}_{2} = \mathbf{A}_t (\mathbf{Z} \bullet \mathbf{A}_s^{\top}), 
\end{align}
where $\bullet$ denotes the row-wise 
Khatri-Rao product.

\subsubsection{Singular value decompositions}
\label{sec:svd}

Throughout this paper, we assume $\min\{R,M,N - M + 1\} \ge L$. From the decompositions~\eqref{S3A} and ~\eqref{S1AsANDS2At}, it follows that $\operatorname{rank}(\mathbf{S}_{1}) \le L$, $\operatorname{rank}(\mathbf{S}_{2}) \le L$, and $\operatorname{rank}(\mathbf{S}_{3}) \le L$.

We denote the compact singular value decompositions (SVDs) of $\mathbf{S}_{1}$, $\mathbf{S}_{2}$, and $\mathbf{S}_{3}$ as $\mathbf{S}_{1} = \mathbf{U}_{\mathbf{S}_1} \mathbf{\Sigma}_{\mathbf{S}_1} \mathbf{V}_{\mathbf{S}_1}^H$, $\mathbf{S}_{2} = \mathbf{U}_{\mathbf{S}_2} \mathbf{\Sigma}_{\mathbf{S}_2} \mathbf{V}_{\mathbf{S}_2}^H$, and $\mathbf{S}_{3} = \mathbf{U}_{\mathbf{S}_3} \mathbf{\Sigma}_{\mathbf{S}_3} \mathbf{V}_{\mathbf{S}_3}^H$, where $\mathbf{\Sigma}_{\mathbf{S}_1}, \mathbf{\Sigma}_{\mathbf{S}_2}, \mathbf{\Sigma}_{\mathbf{S}_3}$ are $L \times L$ diagonal matrices containing the singular values of $\mathbf{S}_1,\mathbf{S}_2,\mathbf{S}_3$ on the diagonal (including zeros if the matrix has rank less than $L$), and $\mathbf{U}_{\mathbf{S}_1}, \mathbf{U}_{\mathbf{S}_2}, \mathbf{U}_{\mathbf{S}_3}$ and
$\mathbf{V}_{\mathbf{S}_1}, \mathbf{V}_{\mathbf{S}_2}, \mathbf{V}_{\mathbf{S}_3}$
are appropriately sized matrices, each with $L$ orthornormal columns.

\section{Methodology}
\label{sec:methodology}

\subsection{Assumptions}
\label{sec:assumptions}

Our analysis rests on four mild assumptions:
\begin{enumerate}
    \item {\bf (Sufficient dimension)} As mentioned in Section~\ref{sec:svd}, we assume $\min\{R,M,N - M + 1\} \ge L$, where $L = \sum_{p=1}^P L_p$ is the total number of harmonic components.
    \item {\bf (Nontrivial amplitudes)} All complex amplitudes $\alpha_{p, l}$ are nonzero.
    \item {\bf (Harmonic incoherence)} The set of all frequencies $\{\Omega_{p,l} = l\omega_p \pmod{2\pi}\}$ contains $L$ unique values.
    \item {\bf (Spatial incoherence)} The set of all spatial frequencies $\{\Psi_{p,l} = l\phi_p \pmod{2\pi}\}$ contains $L$ unique values; recall that $\phi_p = \omega_p f_s \frac{d}{c} \sin(\theta_p)$, so this condition depends both on the DOAs and on the source frequencies.
\end{enumerate}

\begin{proposition}
\label{prop:rank}
Under Assumptions 1--4, the following matrices all have rank $L$: $\mathbf{A}_t$, $\mathbf{A}_s$, $\mathbf{A}$, $\mathbf{Z}$, $\mathbf{S}_{1}$, $\mathbf{S}_{2}$, and $\mathbf{S}_{3}$.
\end{proposition}
\textit{Proof:} Due to its Vandermonde structure, $\mathbf{A}_t$ has rank $L$ if $M \ge L$ and Assumption 3 is satisfied. Similarly, $\mathbf{A}_s$ has rank $L$ if $R \ge L$ and Assumption 4 is satisfied. From~\eqref{eq:AKRprod}, if $\mathbf{A}_t$ and $\mathbf{A}_s$ are rank $L$, then $\mathbf{A}$ must also be rank $L$. The matrix $\mathbf{Z}$ also has a Vandermonde structure and will have rank $L$ if $N-M+1 \ge L$ and Assumptions 2 and 3 hold. Finally, from \eqref{S3A} and \eqref{S1AsANDS2At}, we will have $\operatorname{rank}(\mathbf{S}_{1}) = \operatorname{rank}(\mathbf{S}_{2}) = \operatorname{rank}(\mathbf{S}_{3}) = L$ if $\operatorname{rank}(\mathbf{A}_{t}) = \operatorname{rank}(\mathbf{A}_{s}) = \operatorname{rank}(\mathbf{Z}) = L$. $\hfill\square$

\subsection{Subspace estimation}
\label{sec:subest}

The subspace of $\mathbb{C}^{RM}$ spanned by the columns of $\mathbf{U}_{\mathbf{S}_3}$ contains key information that can be exploited by the ESPRIT method for joint estimation of both DOA and pitch (see Section~\ref{sec:esprit}). We begin by detailing an important connection between this subspace and the mode-1 and mode-2 unfoldings of $\mathcal{S}$. Inspired by \cite[Corollary 1, Appendix A]{roemer2014analytical}, we establish the following relationship between the mode-3 signal subspace and the signal subspaces associated with the other two modes. Although the resulting identity is the same as that in \cite{roemer2014analytical}, our signal model and assumptions differ; therefore, we provide a self-contained proof tailored to the present setting. 

\begin{proposition}
\label{prop:main}
Suppose Assumptions 1--4 above are met. Define $\mathbf{T}_1 = \mathbf{U}_{\mathbf{S}_1} \mathbf{U}_{\mathbf{S}_1}^H$ and $\mathbf{T}_2 = \mathbf{U}_{\mathbf{S}_2} \mathbf{U}_{\mathbf{S}_2}^H$. Then
\begin{equation}
\label{prop1}
    \mathbf{U}_{\mathbf{S}_3} = (\mathbf{T}_2 \otimes \mathbf{T}_1) \mathbf{U}_{\mathbf{S}_3}.
\end{equation}
\end{proposition}
\textit{Proof:} Defining $\mathbf{P}_t := \mathbf{A}_t \mathbf{A}_t^+$ and $\mathbf{P}_s := \mathbf{A}_s \mathbf{A}_s^+$ as orthogonal projectors onto the column span of $\mathbf{A}_t$ and $\mathbf{A}_s$, respectively, it follows that $\mathbf{A}_t = \mathbf{P}_t \mathbf{A}_t$ and $\mathbf{A}_s = \mathbf{P}_s \mathbf{A}_s$. Taking the Khatri-Rao product of $\mathbf{A}_t$ and $\mathbf{A}_s$, we obtain $\mathbf{A} = \mathbf{A}_t * \mathbf{A}_s = (\mathbf{P}_t \mathbf{A}_t) * (\mathbf{P}_s \mathbf{A}_s) = (\mathbf{P}_t \otimes \mathbf{P}_s) (\mathbf{A}_t * \mathbf{A}_s) = (\mathbf{P}_t \otimes \mathbf{P}_s) \mathbf{A}$.

Using~\eqref{S3A} and~\eqref{S1AsANDS2At}, we then have
\begin{align*}
\mathbf{S}_{1} &= \mathbf{A}_s (\mathbf{Z} \bullet \mathbf{A}_t^{\top}) = \mathbf{P}_s \mathbf{A}_s (\mathbf{Z} \bullet \mathbf{A}_t^{\top}) = \mathbf{P}_s \mathbf{S}_{1}, \\
\mathbf{S}_{2} &= \mathbf{A}_t (\mathbf{Z} \bullet \mathbf{A}_s^{\top}) = \mathbf{P}_t \mathbf{A}_t (\mathbf{Z} \bullet \mathbf{A}_s^{\top}) = \mathbf{P}_t \mathbf{S}_{2},~\text{and} \\
\mathbf{S}_3 &= \mathbf{A}\mathbf{Z} = (\mathbf{P}_t \otimes \mathbf{P}_s) \mathbf{A}\mathbf{Z} = (\mathbf{P}_t \otimes \mathbf{P}_s) \mathbf{S}_3.
\end{align*}

From Proposition~\ref{prop:rank}, $\mathbf{A}_t$, $\mathbf{A}_s$, $\mathbf{A}$, $\mathbf{Z}$, $\mathbf{S}_{1}$, $\mathbf{S}_{2}$, and $\mathbf{S}_{3}$ all have rank $L$. Thus $\mathbf{S}_{1} = \mathbf{U}_{\mathbf{S}_1} \mathbf{\Sigma}_{\mathbf{S}_1} \mathbf{V}_{\mathbf{S}_1}^H = \mathbf{A}_s (\mathbf{Z} \bullet \mathbf{A}_t^{\top})$,
from which it follows that
$\mathbf{U}_{\mathbf{S}_1} = \mathbf{A}_s (\mathbf{Z} \bullet \mathbf{A}_t^{\top}) \mathbf{V}_{\mathbf{S}_1} (\mathbf{\Sigma}_{\mathbf{S}_1})^{-1} =: \mathbf{A}_s \mathbf{R}_1$,
where $\mathbf{\Sigma}_{\mathbf{S}_1}, \mathbf{R}_1 \in \mathbb{C}^{L \times L}$ are both invertible. Similarly, $\mathbf{U}_{\mathbf{S}_2} = \mathbf{A}_t (\mathbf{Z} \bullet \mathbf{A}_s^{\top}) \mathbf{V}_{\mathbf{S}_2} (\mathbf{\Sigma}_{\mathbf{S}_2})^{-1} =: \mathbf{A}_t \mathbf{R}_2$ 
and 
$\mathbf{U}_{\mathbf{S}_3} = \mathbf{A} \mathbf{Z} \mathbf{V}_{\mathbf{S}_3} (\mathbf{\Sigma}_{\mathbf{S}_3})^{-1} =: \mathbf{A} \mathbf{R}_3$, 
where $\mathbf{\Sigma}_{\mathbf{S}_2},\mathbf{\Sigma}_{\mathbf{S}_3}, \mathbf{R}_2,\mathbf{R}_3 \in \mathbb{C}^{L \times L}$ are all invertible. Hence,
\begin{equation*}
\begin{aligned}
    \mathbf{P}_s &:= \mathbf{A}_s \mathbf{A}_s^+ = (\mathbf{U}_{\mathbf{S}_1} \mathbf{R}_1^{-1}) (\mathbf{U}_{\mathbf{S}_1} \mathbf{R}_1^{-1})^+ =\mathbf{U}_{\mathbf{S}_1} \mathbf{R}_1^{-1} \mathbf{R}_1 \mathbf{U}_{\mathbf{S}_1}^+ \\
    &= \mathbf{U}_{\mathbf{S}_1} \mathbf{U}_{\mathbf{S}_1}^+ = \mathbf{U}_{\mathbf{S}_1} (\mathbf{U}_{\mathbf{S}_1}^H \mathbf{U}_{\mathbf{S}_1})^{-1} \mathbf{U}_{\mathbf{S}_1}^H 
    = \mathbf{U}_{\mathbf{S}_1} \mathbf{U}_{\mathbf{S}_1}^H = \mathbf{T}_1,
\end{aligned}
\end{equation*}
where the second to the last equality uses the fact $\mathbf{U}_{\mathbf{S}_1}^H \mathbf{U}_{\mathbf{S}_1} = \mathbf{I}$. Similarly, we can obtain $\mathbf{P}_t := \mathbf{A}_t \mathbf{A}_t^+ = \mathbf{U}_{\mathbf{S}_2} \mathbf{U}_{\mathbf{S}_2}^H = \mathbf{T}_2$.

Finally, we have $(\mathbf{T}_2 \otimes \mathbf{T}_1) \mathbf{U}_{\mathbf{S}_3} 
= (\mathbf{P}_t \otimes \mathbf{P}_s) \mathbf{U}_{\mathbf{S}_3} = (\mathbf{P}_t \otimes \mathbf{P}_s) \mathbf{A} \mathbf{R}_3 = \mathbf{A} \mathbf{R}_3 = \mathbf{U}_{\mathbf{S}_3}$, which proves~\eqref{prop1}. $\hfill\square$

The matrix $\mathbf{T}_2 \otimes \mathbf{T}_1$ appearing in~\eqref{prop1} is itself an orthogonal projection matrix. It is self-adjoint (as $(\mathbf{T}_2 \otimes \mathbf{T}_1)^H = \mathbf{T}_2^H \otimes \mathbf{T}_1^H = \mathbf{T}_2 \otimes \mathbf{T}_1$) and idempotent (as 
$(\mathbf{T}_2 \otimes \mathbf{T}_1)(\mathbf{T}_2 \otimes \mathbf{T}_1) = (\mathbf{T}_2 \mathbf{T}_2) \otimes (\mathbf{T}_1 \mathbf{T}_1) = \mathbf{T}_2 \otimes \mathbf{T}_1$).

In practice, $\mathbf{U}_{\mathbf{S}_3}$ is not directly available and must be estimated from noisy data in $\mathcal{X}$. A classical method for doing so is to compute the SVD of the mode-3 unfolding of $\mathcal{X}$: $\mathbf{X}_3 = \mathbf{U}_{\mathbf{X}_3} \mathbf{\Sigma}_{\mathbf{X}_3} \mathbf{V}_{\mathbf{X}_3}^H$.
Because $\mathbf{X}_3$ may be full rank in general, the subspace estimate can be obtained from the span of its first $L$ singular vectors. We denote this as 
\begin{equation}
\Umat := \mathbf{U}_{\mathbf{X}_3}(:,1:L)
\label{eq:usmat}
\end{equation}
because this estimate can be formed directly from the matrix $\mathbf{X}_3 = [\mathrm{vec}(\mathbf{X}(0)) ~ \mathrm{vec}(\mathbf{X}(1)) ~ \cdots ~ \mathrm{vec}(\mathbf{X}(N-M))]$ without using tensors. We also note that $\Umat$ is the same subspace estimate employed in~\cite{wu2014joint}.

The insight to be gleaned from Proposition~\ref{prop:main}, however, is that additional subspace structure can be revealed through multiple unfoldings of the data tensor $\mathcal{X}$. Specifically, $\mathbf{T}_1$ and $\mathbf{T}_2$ can be estimated from the SVD of the mode-1 and mode-2 unfoldings of $\mathcal{X}$ as follows: 
\begin{equation}
\widehat{\mathbf{T}}_1 := \mathbf{U}_{\mathbf{X}_1}(:,1:L) (\mathbf{U}_{\mathbf{X}_1}(:,1:L))^H
\label{eq:t1hat}
\end{equation}
and
\begin{equation}
\widehat{\mathbf{T}}_2 := \mathbf{U}_{\mathbf{X}_2}(:,1:L) (\mathbf{U}_{\mathbf{X}_2}(:,1:L))^H.
\label{eq:t2hat}
\end{equation}
With the estimated projection operator $\widehat{\mathbf{T}}_2 \otimes \widehat{\mathbf{T}}_1$, we can then refine the matrix-based subspace estimate to obtain our tensor-based subspace estimate: 
\begin{equation}
\Uten := (\widehat{\mathbf{T}}_2 \otimes \widehat{\mathbf{T}}_1) \Umat. 
\label{eq:usten}
\end{equation}
While the columns of $\Uten$ are not automatically orthonormal, its column span $\range(\Uten)$ is the refined subspace estimate. By projecting the matrix-based subspace estimate $\Umat$ onto the subspace constrained by both the spatial and temporal steering structures, we effectively suppress noise components that do not conform to the harmonic steering structure defined in \eqref{eq:AKRprod}. The goal is for improvement in the subspace estimate, provided by the approximate knowledge of the orthogonal projector $\mathbf{T}_2 \otimes \mathbf{T}_1$, to translate into improved pitch and DOA accuracy. 

In practice, it is possible for the empirical tensor-based subspace estimate $\Uten$ to be slightly worse than $\Umat$, but this is a rare occurrence in our simulations (see Section~\ref{sec:examples}). Much more commonly, $d(\Uten, \mathbf{U}_{\mathbf{S}_3}) < d(\Umat, \mathbf{U}_{\mathbf{S}_3})$. Section~\ref{sec:guarantee} of this paper is devoted to proving that $d(\Uten, \mathbf{U}_{\mathbf{S}_3}) < d(\Umat, \mathbf{U}_{\mathbf{S}_3})$ with high probability in a variety of scenarios.

\subsection{Pitch and DOA estimation}
\label{sec:esprit}

Once the subspace estimation is finished, we can then proceed to use the ESPRIT method for joint estimation of DOA and pitch. From the proof of Proposition~\ref{prop:main}, we have that $\mathbf{U}_{\mathbf{S}_3} = \mathbf{A} \mathbf{R}_3$ where $\mathbf{R}_3$ is invertible, so $\mathbf{U}_{\mathbf{S}_3}$ and $\mathbf{A}$ share the same column span. Following~\cite{wu2014joint}, we exploit the shift invariance structure in $\mathbf{A}$ to estimate DOA and pitch.

We start by defining the following four selection matrices
\begin{align*}
\mathbf{W}_1 &= 
\begin{bmatrix}
\mathbf{I}_{M-1} \\
\mathbf{0}_{1\times(M-1)}
\end{bmatrix}
\otimes \mathbf{I}_R, 
&\mathbf{W}_2 = 
\begin{bmatrix}
\mathbf{0}_{1\times(M-1)} \\
\mathbf{I}_{M-1}
\end{bmatrix}
\otimes \mathbf{I}_R, \\[6pt]
\mathbf{W}_3 &= 
\mathbf{I}_M \otimes
\begin{bmatrix}
\mathbf{I}_{R-1} \\
\mathbf{0}_{1\times(R-1)}
\end{bmatrix}, 
&\mathbf{W}_4 = 
\mathbf{I}_M \otimes
\begin{bmatrix}
\mathbf{0}_{1\times(R-1)} \\
\mathbf{I}_{R-1}
\end{bmatrix}.
\end{align*}
The selection matrices $\mathbf{W}_1$ and $\mathbf{W}_2$ have size $MR \times (M-1)R$, while $\mathbf{W}_3$ and $\mathbf{W}_4$ have size $MR \times M(R-1)$. 

Using the selection matrices, we have the following shift invariant relationships involving various submatrices of $\mathbf{A}$: $\mathbf{W}_1^T \mathbf{A} \boldsymbol{\Psi}_t = \mathbf{W}_2^T \mathbf{A}$ and $\mathbf{W}_3^T \mathbf{A}  \boldsymbol{\Psi}_s = \mathbf{W}_4^T \mathbf{A}$, where $\boldsymbol{\Psi}_t$ and $\boldsymbol{\Psi}_s$ are $L \times L$ matrices defined as $\boldsymbol{\Psi}_t = \operatorname{diag}\!\left(\boldsymbol{\Psi}_{t_1}, \ldots, \boldsymbol{\Psi}_{t_P}\right)$ where
$\boldsymbol{\Psi}_{t_p} = \operatorname{diag}\!\left(e^{j\omega_p}, \ldots, e^{j\omega_p L_p}\right)$
and $\boldsymbol{\Psi}_s = \operatorname{diag}\!\left(\boldsymbol{\Psi}_{s_1}, \ldots, \boldsymbol{\Psi}_{s_P}\right)$ where
$\boldsymbol{\Psi}_{s_p} = \operatorname{diag}\!\left(e^{j\phi_p}, \ldots, e^{j\phi_p L_p}\right)$.

Substituting $\mathbf{A}=\mathbf{U}_{\mathbf{S}_3}\mathbf{R}_3^{-1}$ into the above, we obtain the relations
$\mathbf{W}_1^T \mathbf{U}_{\mathbf{S}_3} \boldsymbol{\Phi}_t  = \mathbf{W}_2^T \mathbf{U}_{\mathbf{S}_3}$ and $\mathbf{W}_3^T \mathbf{U}_{\mathbf{S}_3} \boldsymbol{\Phi}_s = \mathbf{W}_4^T \mathbf{U}_{\mathbf{S}_3}$, which involve the similarity transformations $\boldsymbol{\Phi}_t := \mathbf{R}_3^{-1} \boldsymbol{\Psi}_t \mathbf{R}_3 $ and $\boldsymbol{\Phi}_s := \mathbf{R}_3^{-1} \boldsymbol{\Psi}_s \mathbf{R}_3$. 
It follows that the eigenvalues of $\boldsymbol{\Phi}_t$ and $\boldsymbol{\Phi}_s$ are the same as the diagonal entries of $\boldsymbol{\Psi}_t$ and $\boldsymbol{\Psi}_s$; we denote these by $\mu_l^{(p)}$ and $\nu_l^{(p)}$, respectively, for $l=1,\ldots,L_p$ and $p=1,\ldots,P$. Recall that $\operatorname{rank}(\mathbf{U}_{\mathbf{S}_3})=L$. If  $\operatorname{rank}(\mathbf{W}_1^T\mathbf{U}_{\mathbf{S}_3})=L$ and $\operatorname{rank}(\mathbf{W}_3^T\mathbf{U}_{\mathbf{S}_3})=L$, it follows that $\boldsymbol{\Phi}_t  = (\mathbf{W}_1^T \mathbf{U}_{\mathbf{S}_3})^+ \mathbf{W}_2^T \mathbf{U}_{\mathbf{S}_3}$  and $
\boldsymbol{\Phi}_s = (\mathbf{W}_3^T \mathbf{U}_{\mathbf{S}_3})^+ \mathbf{W}_4^T \mathbf{U}_{\mathbf{S}_3}$, where $^+$ stands for the Moore-Penrose pseudoinverse. 

In practice, $\mathbf{U}_{\mathbf{S}_3}$ is not directly available, but one of its estimates---either $\Umat$ or $\Uten$ from Section~\ref{sec:subest}---can be used in its place to compute estimates of $\boldsymbol{\Phi}_t$ and $\boldsymbol{\Phi}_s$. The eigenvalues of these estimates are denoted by $\widehat{\mu}_l^{(p)}$ and $\widehat{\nu}_l^{(p)}$, respectively. Because the principal argument is defined modulo $2\pi$, the phases are first unwrapped across the harmonic index $l$; hereafter, $\angle(\cdot)$ denotes the corresponding unwrapped phase, chosen such that
\[
\angle\widehat{\mu}_l^{(p)} \approx l\omega_p
\qquad \text{and} \qquad
\angle\widehat{\nu}_l^{(p)} \approx l\phi_p.
\]
The pitch and DOA parameters can then be estimated by averaging these phases over the model order as
\[
\widehat{\omega}_p
=
\frac{2}{L_p(L_p+1)}
\sum_{l=1}^{L_p}
\angle\widehat{\mu}_l^{(p)}
\]
and
\begin{align*}
\widehat{\theta}_p
&=
\arcsin\!\left(
\frac{2c}{L_p(L_p+1)\widehat{\omega}_p f_s d}
\sum_{l=1}^{L_p}
\angle\widehat{\nu}_l^{(p)}
\right).
\end{align*}
In the multi-pitch case, the estimated eigenvalues $\widehat{\mu}_l^{(p)}$ and $\widehat{\nu}_l^{(p)}$ must be correctly associated to ensure consistent pairing between pitch and DOA estimates. Standard eigenvalue pairing techniques available in the literature (e.g.,~\cite{duofang2008angle, liu1998azimuth}) can be employed to accomplish this task.

\section{Theoretical Guarantee for Tensor Gain}
\label{sec:guarantee}

In this section, we prove that in a variety of scenarios, $d(\Uten, \mathbf{U}_{\mathbf{S}_3}) < d(\Umat, \mathbf{U}_{\mathbf{S}_3})$ with high probability. To do this, we decompose the overall tensor gain $d(\Umat,\Utrue)-d(\Uten,\Utrue)$ into two components as detailed in equation~\eqref{eq:overalltg2}. We provide deterministic bounds on these components in Theorem~\ref{thm:oraclebounds} and Theorem~\ref{thm:exact}. As these bounds involve factors (such as $\Umat$) that ultimately depend on the noise tensor $\mathcal{Q}$, we then establish probabilistic bounds on the relevant quantities. Our analysis culminates in our main result, Theorem~\ref{thm:main}, which establishes that in certain scenarios, the overall tensor gain is positive with high probability. 

Our analysis is entirely non-asymptotic, as it can be applied to finite values of $R$, $M$, $N$, and SNR, and all constants (such as in the failure probabilities) are explicit. In contrast to the asymptotic framework of \cite{roemer2014analytical}, our results provide finite-sample conditions under which a strictly positive tensor gain is guaranteed for a fixed problem configuration.

\subsection{Setup}
\label{sec:proofsetup}

As in Proposition~\ref{prop:main}, suppose Assumptions 1--4 are met. Therefore, from Proposition~\ref{prop:rank}, the following matrices all have rank $L$: $\mathbf{A}_t$, $\mathbf{A}_s$, $\mathbf{A}$, $\mathbf{Z}$, $\mathbf{S}_{1}$, $\mathbf{S}_{2}$, and $\mathbf{S}_{3}$.

As in Section~\ref{sec:sigmod}, we suppose the base noise matrix $\mathbf{Q} \in \mathbb{C}^{R \times N}$ has i.i.d.\ entries distributed as $\mathcal{CN}(0, \sigma^2)$.

Recall that the columns of  $\Utrue \in \mathbb{C}^{RM \times L}$ define an orthonormal basis for the true $L$-dimensional mode-3 signal subspace, that $\Umat \in \mathbb{C}^{RM \times L}$ defines an orthonormal basis for the matrix-based estimate of the subspace as defined in~\eqref{eq:usmat}, and that $\Uten \in \mathbb{C}^{RM \times L}$ defines a basis for our tensor-based estimate of the subspace as defined in~\eqref{eq:usten}.

Recall also the projectors $\mathbf T_1=\mathbf U_{S_1}\mathbf U_{S_1}^H$ and $\mathbf T_2=\mathbf U_{S_2}\mathbf U_{S_2}^H$ together with their estimates $\widehat{\mathbf T}_1$ (see~\eqref{eq:t1hat}) and $\widehat{\mathbf T}_2$ (see~\eqref{eq:t2hat}) from the mode-1 and mode-2 unfoldings of the data tensor $\boldsymbol{\mathcal X}$. 

We let $\Ksub := \range(\mathbf T_2\otimes\mathbf T_1)$ denote the true (``oracle'') subspace constrained by the spatial and temporal steering structures, and we let $\Ksubhat := \range(\widehat{\mathbf T}_2\otimes\widehat{\mathbf T}_1)$ denote its empirical estimate. We define the corresponding orthogonal projection operators onto $\Ksub$ and $\Ksubhat$ as 
\[
\PK := \mathbf T_2\otimes\mathbf T_1 \quad \text{and} \quad
\PKhat :=\widehat{\mathbf T}_2\otimes\widehat{\mathbf T}_1,
\]
and we refer to $\PK$ and $\PKhat$ as the oracle and empirical subspace projectors, respectively. By Proposition~\ref{prop:main}, 
\begin{equation}
\PK\,\Utrue=\Utrue. \label{eq:fixed}
\end{equation}

Due to~\eqref{eq:usten}, we have that $\Uten = \PKhat \Umat$; that is, the empirical subspace projector is used to refine the matrix-based subspace estimate from~\eqref{eq:usmat}. We also define
\begin{equation}
    \Uorc := \PK\Umat;
    \label{eq:Uorc}
\end{equation} 
its range is the refined subspace estimate that would hypothetically be possible with knowledge of the oracle subspace projector~$\PK$.

\subsection{Decomposing overall tensor gain}

Using the concept of distance between subspaces, we define the {\bf overall tensor gain} as the improvement in estimating $\mathbf U_{S_3}$ via our refined empirical estimate $\Uten$ compared to the the original matrix-based estimate $\Umat$:
\[
\text{overall tensor gain} 
:= d(\Umat,\Utrue)-d(\Uten,\Utrue).
\]
Both analytically and intuitively, we find it useful to decompose the overall tensor gain into the hypothetical gain that would be possible with the oracle estimate $\Uorc$ (see~\eqref{eq:Uorc}), minus the amount of gain that is ``given back'' due to the use of the empirical (rather than oracle) projector. That is, 
\begin{equation}
\label{eq:overalltg2}
\text{overall tensor gain} 
= \gor - \lem,
\end{equation}
where we define the {\bf oracle gain} and {\bf empirical loss} as follows:
\begin{align*}
\gor &:= d(\Umat,\Utrue)-d(\Uorc,\Utrue), \\
\lem &:= d(\Uten,\Utrue)-d(\Uorc,\Utrue).
\end{align*}
Our ultimate goal is to prove that the overall tensor gain is positive with high probability in a variety of scenarios. We approach this by identifying cases where the oracle gain exceeds the empirical loss.

\subsection{Oracle gain}
\label{sec:oraclegain}

Recall that the oracle gain $\gor = d(\Umat,\Utrue)-d(\Uorc,\Utrue)$ is the hypothetical tensor gain that would be possible using the oracle subspace projector $\PK$ compared to using the empirical subspace projector $\PKhat$. As we establish in the following lemma, the oracle gain can never be negative. Proofs of this proposition and all other theoretical results from Section~\ref{sec:guarantee} are provided in the Appendix (see Supplementary Material).

\begin{proposition}[Oracle gain is nonnegative]
\label{prop:oracle}
$\gor \ge 0$.
\end{proposition}

The proof of Proposition~\ref{prop:oracle} relies on the following lemma.

\begin{lemma}[Retained-energy floor]
\label{lem:floor}
If $d(\Umat,\Utrue)<1$ then $\PK\Umat$ has rank $L$ and $\rho > 0$, where 
\begin{equation}
    \rho := \sigmaL(\PK\Umat). \label{eq:rho} 
\end{equation}
\end{lemma}
The condition $d(\Umat,\Utrue)<1$ is satisfied except in pathological cases where the noise in $\mathbf{X}_3 = \mathbf{S}_3 + \mathbf{Q}_3$ causes its first $L$ principal components to be completely orthogonal to those of $\mathbf{S}_3$. We henceforth adopt a fifth assumption.
\begin{enumerate}
\setcounter{enumi}{4} 
\item {\bf (Nonzero correlation)} The matrix estimate $\Umat$ has nonzero correlation with the ground truth subspace, i.e., $d(\Umat,\Utrue)<1$. Similarly, we assume that $\range(\PKhat\Umat)$ is $L$-dimensional.
\end{enumerate}

We can furthermore establish deterministic lower and upper bounds on the oracle gain. To do this, we define
\begin{align}
a &:= \sigmaL^2(\Utrue^{\herm}\Uhat), \label{eq:adef}
\\
\Epar &:= (\PK-\PiS)\,\Uhat.  \label{eq:epardef}
\end{align}

\begin{theorem}[Deterministic bounds on oracle gain]
\label{thm:oraclebounds}
The oracle gain satisfies 
\[
\gorlower \le \gor \le \gorupper,
\]
where
\[
\gorlower := \sqrt{1-a}\;-\;\frac{\sqrt{\|\Epar\|_2^2}}{\sqrt{a+\|\Epar\|_2^2}}.
\]
and
\[
\gorupper := \sqrt{1-a} - \sqrt{1-\frac{a}{\rho^2}}.
\]
\end{theorem}

Theorem~\ref{thm:oraclebounds} warrants some interpretation. First, we note that in situations where $d(\Umat,\Utrue)$ is small to begin with, there is not much room for oracle gain. This corresponds to situations where $a$ is close to $1$, causing the first square root term in $\gorupper$ to be near $0$. Second, we note that if the largest principal angle between $\range(\Umat)$ and $\Ksub$ is near $0$, then $\rho$ will be close to $1$, and there will be little separation between the two square root terms in the upper bound. Note that because of~\eqref{eq:fixed}, if there exists at least one large angle between $\range(\Umat)$ and $\Ksub$, there must exist at least one large angle between $\range(\Umat)$ and $\range(\Utrue)$. On the other hand, if the largest principal angle between $\range(\Umat)$ and $\range(\Utrue)$ is large, $a$ will be close to $0$, which shrinks $\gorlower$. So there is a ``sweet spot'' for $a$; we revisit this in the numerical experiments. Third, in situations where $\|\Epar\|_2$ is large, $\gorlower$ will be suppressed. In our probabilistic analysis, we identify situations where $\|\Epar\|_2$ can be bounded from above with high probability. 

\subsection{Empirical loss}
\label{sec:emploss}

Recall that the empirical loss $\lem = d(\Uten,\Utrue)-d(\Uorc,\Utrue)$ quantifies how much of the oracle gain is ``given back'' due to the error in estimating $\Ksub$. 

Our next theorem establishes a deterministic bound on the empirical loss. To do this, we define
\begin{equation}
\eta := \spec{\PKhat-\PK}.
\label{eq:etadef}
\end{equation}

\begin{theorem}[Deterministic bound on empirical loss]
\label{thm:exact}
If $\eta < \rho$, the empirical loss satisfies
\[
\lem \le \lemupper,
\]
where
\[
\lemupper := \frac{\eta}{\rho-\eta}.
\]
\end{theorem}

Combining Theorems~\ref{thm:oraclebounds} and~\ref{thm:exact} with~\eqref{eq:overalltg2}, we see that in situations where $\lemupper < \gorlower$, we are guaranteed to have positive overall tensor gain. While this condition does not hold deterministically, we identify a range of scenarios where it holds with high probability.

\subsection{Bounding terms in $\gorlower$ and $\lemupper$}
\label{sec:prob}

In this section, we focus on bounding $\rho$, $a$, $\|\Epar\|_2$, and $\eta$, which are the terms appearing in $\lemupper$ and $\gorlower$.

~

\subsubsection{Bounding $\rho$ from below} 

To help bound $\lemupper$ from above, we bound $\rho$ from below. Recall from~\eqref{eq:rho} that $\rho = \sigmaL(\PK\Umat)$.

\begin{lemma}
[Lower bound on $\rho$]
\label{lem:rholower}
$\rho^2 \ge a$.
\end{lemma}

We establish a lower bound on $a$ in a subsequent section. 

~

\subsubsection{Bounding $\eta$ from above}

To help bound $\lemupper$ from above, we also bound $\eta$ from above. Recall from~\eqref{eq:etadef} that the joint projection error $\eta = \spec{\PKhat-\PK}$.

The following lemma decouples the joint Kronecker error into individual mode-space errors.

\begin{lemma}[Kronecker telescoping bound]
\label{lem:kronecker_tele}
Let $\mathbf T_i$ and $\widehat{\mathbf T}_i$ ($i=1,2$) be the orthogonal projection matrices onto the true and estimated mode-$i$ signal subspaces. Define $\delta_i := \spec{\widehat{\mathbf T}_i-\mathbf T_i}$ for $i=1,2$. Then, the distance between the joint tensor projectors satisfies:
\[
\eta \le \delta_1 + \delta_2.
\]
\end{lemma}

By utilizing standard subspace perturbation theory, we can explicitly relate each $\delta_i$ directly to the underlying signal-to-noise dynamics of the individual unfoldings.

\begin{corollary}[Perturbation bound for $\eta$]
\label{cor:wedin_sin}
Let $\mathbf{X}_i = \mathbf{S}_i + \mathbf{Q}_i$ represent the mode-$i$ unfoldings of the noisy data tensor $\boldsymbol{\mathcal{X}}$, where $\mathbf{S}_i$ has rank $L$. Define $\gamma_i := \sigmaL(\mathbf S_i)$ for $i=1,2$. If the spectral norm of the unfolded noise matrix satisfies $\spec{\mathbf{Q}_i} < \gamma_i$ for $i=1,2$, then:
\begin{equation}
\eta \le \frac{\spec{\mathbf{Q}_1}}{\gamma_1 - \spec{\mathbf{Q}_1}} + \frac{\spec{\mathbf{Q}_2}}{\gamma_2 - \spec{\mathbf{Q}_2}}.
\label{eq:etaupperQ1Q2}
\end{equation}
\end{corollary}

We next derive non-asymptotic concentration bounds for the spectral norms $\spec{\mathbf{Q}_1}$ and $\spec{\mathbf{Q}_2}$ appearing in~\eqref{eq:etaupperQ1Q2}. We first bound (with high probability) the spectral norm of the mode-1 noise unfolding. To state this result, we define
\[
K := N - M + 1.
\]

\begin{lemma}[Mode-1 noise concentration]
\label{lem:mode1noise} For any $t_1 > 0$, the spectral norm of the mode-1 noise unfolding satisfies:
\begin{align*}
& \mathbb{P}\left( \spec{\mathbf{Q}_1} > \sigma \sqrt{2\min(M, K)} \left( \sqrt{R} + \sqrt{N} \right) + t_1 \right) \\
& \quad \le \exp\left( -\frac{t_1^2}{\sigma^2 \min(M, K)} \right).
\end{align*}
\end{lemma}

Next we bound (with high probability) the spectral norm of the mode-2 noise unfolding.

\begin{lemma}[Mode-2 noise concentration]
\label{lem:mode2noise}
For any $t_2 > 0$, the spectral norm of the mode-2 noise unfolding satisfies:
\begin{align*}
& \mathbb{P}\left( \spec{\mathbf{Q}_2} > \sigma \sqrt{2K} (\sqrt{R} + \sqrt{N}) + t_2 \right) \\
& \quad \le \exp\left( -\frac{t_2^2}{\sigma^2 \min(M, K)} \right).
\end{align*}
\end{lemma}

Combining Lemma~\ref{lem:mode1noise} and Lemma~\ref{lem:mode2noise} with Corollary~\ref{cor:wedin_sin} provides a probabilistic upper bound for the joint projection error $\eta$.

~

\subsubsection{Bounding $a$ from below}

To help with bounding both $\gorlower$ and $\lemupper$, we establish a lower bound on $a$. Recall from~\eqref{eq:adef} that $a = \sigmaL^2(\Utrue^{\herm}\Uhat)$.

\begin{lemma}[Lower bound on baseline subspace alignment]
\label{lemma:baseline_alignment_lower}
Let $\Utrue \in \mathbb{C}^{RM \times L}$ be the true mode-3 signal subspace basis, and let $\Umat \in \mathbb{C}^{RM \times L}$ be the baseline matrix-based estimate obtained from the $L$ dominant left singular vectors of the mode-3 unfolding $\mathbf{X}_3 = \mathbf{S}_3 + \mathbf{Q}_3$. Let $\gamma_3 = \sigmaL(\mathbf{S}_3)$ denote the singular value gap of the true signal matrix. If $\|\mathbf{Q}_3\|_2 < \frac{1}{2}\gamma_3$, then
\begin{equation}
    a \ge 1 - \left( \frac{\|\mathbf{Q}_3\|_2}{\gamma_3 - \|\mathbf{Q}_3\|_2} \right)^2.
\end{equation}
\end{lemma}

Next we bound the spectral norm of the mode-3 noise unfolding.

\begin{lemma}[Mode-3 noise concentration]
\label{lemma:mode3_noise_concentration}
For any tail parameter $t_3 > 0$, the spectral norm of the mode-3 noise unfolding satisfies:
\begin{align*}
& \mathbb{P}\left( \|\mathbf{Q}_3\|_2 > \sigma \sqrt{2M} \left( \sqrt{R} + \sqrt{N} \right) + t_3 \right) \\
& \quad \le \exp\left( -\frac{t_3^2}{\sigma^2 \min(M, K)} \right).
\end{align*}
\end{lemma}

Combining Lemma~\ref{lemma:baseline_alignment_lower} with Lemma~\ref{lemma:mode3_noise_concentration} provides a probabilistic lower bound for $a$.

~

\subsubsection{Bounding $a$ from above}

To help with bounding $\gorlower$ from below, we also establish an upper bound on $a = \sigmaL^2(\Utrue^{\herm}\Uhat)$.

For this analysis, let $\Uperp\in\mathbb{C}^{RM\times(RM-L)}$ be an orthonormal basis for
$\range(\Utrue)^{\perp}$. Let $\Uhat,\Vhat$ hold the top-$L$ left/right singular vectors of $\Xt$,
with $\uLh,\vLh$ their $L$-th columns and $\widehat\sigma_L:=\sigmaL(\Xt)$. Let $\vL$ be the
$L$-th right singular vector of $\St$.  

\begin{lemma}[Upper bound on baseline subspace alignment]\label{lem:det}
If $\|\Qt\|_2<\gamma_3$, then
\begin{equation*}
a \le 1-\!\left(\frac{\big\|\UperpH\Qt\,\vLh\big\|_2}{\gamma_3+\|\Qt\|_2}\right)^{2}.    
\end{equation*}
\end{lemma}

The norm $\|\mathbf{Q}_3\|_2$ can be bounded probabilistically using Lemma~\ref{lemma:mode3_noise_concentration}. We also probabilistically bound $\big\|\UperpH\Qt\,\vLh\big\|_2$, which appears in the upper bound on $a$.

\begin{lemma}[Probabilistic floor on the weakest-direction leakage]\label{lem:prob}
Suppose $K\le RM-L$. Fix $\varepsilon\in(0,\tfrac12)$ and  $\mu\in(0,1)$, and define
\begin{equation}
p_{\mathrm{w}}(\varepsilon,\mu)
:=2\Big(1+\tfrac{2}{\varepsilon}\Big)^{2K}
\exp\!\Big(-\frac{\mu^2(1-2\varepsilon)^2\,RM}{4\min(M,K)}\Big).
\label{eq:pw}
\end{equation}
Then with probability at least
$1-p_{\mathrm{w}}(\varepsilon,\mu)$,
\begin{equation}
\bigl\|\UperpH\mathbf Q_3\widehat{\mathbf v}_L\bigr\|_2
\;\ge\;
\sigma\sqrt{\bigl[\,RM(1-\mu)-L\min(M,K)\,\bigr]_+}\,.
\label{star_mu}
\end{equation}
\end{lemma}

We note that the failure probability $p_{\mathrm{w}}(\varepsilon,\mu)$ in~\eqref{eq:pw} is minimized by choosing $\varepsilon \approx 1/16$. This leads us to a probabilistic upper bound for $a$. 

~

\subsubsection{Bounding $\spec{\Epar}^2$ from above}

To help with bounding $\gorlower$ from below, we establish an upper bound on $\spec{\Epar}^2$. Recall from~\eqref{eq:epardef} that $\Epar = (\PK-\PiS)\,\Uhat$.

\begin{lemma}[In-cage identity]\label{lem:incage}
Let $\Xt=\St+\Qt$ with $\St$ of rank $L$ and $\gamma_3=\sigmaL(\St)$, and let
$\Uhat,\Vhat,\Shat$ be the top-$L$ left singular vectors, right singular vectors, and singular
values of $\Xt$. If $\twonorm{\Qt}<\gamma_3$ then
\begin{equation}
\twonorm{\Epar}^2\;\le\;\left(\frac{\twonorm{(\PK-\PiS)\,\Qt}}{\gamma_3-\twonorm{\Qt}}\right)^2.
\label{eq:incage-bound}
\end{equation}
\end{lemma}

The norm $\|\mathbf{Q}_3\|_2$ can be bounded probabilistically using Lemma~\ref{lemma:mode3_noise_concentration}. We also probabilistically bound $\twonorm{(\PK-\PiS)\,\Qt}$.

\begin{lemma}[Probabilistic in-cage envelope]\label{lem:incageenv}
For any $t_4>0$,
\begin{align}
& \mathbb{P}\left( \twonorm{(\PK-\PiS)\Qt} > \sigma\sqrt{(L^2-L)K}+t_4 \right) \nonumber \\
& \le \exp\left( -\frac{t_4^2}{\sigma^2 \min(M, K)} \right).
\label{eq:incage-prob}
\end{align}
\end{lemma}

This leads us to a probabilistic upper bound for $\spec{\Epar}^2$.

\subsection{Proving that $\lemupper < \gorlower$ with high probability}
\label{sec:finalproof}

We are now in position to prove our main result: that in certain scenarios the overall tensor gain is positive with high probability.

\begin{theorem}[Main result]\label{thm:main}
Let $\gamma_i := \sigmaL(\mathbf S_i)$ for $i=1,2,3$ be the unfolding gaps. Fix positive slacks $t_1,t_2,t_3,t_4 > 0$ and define 
\begin{align*}
\omega_1&:=\sigma\sqrt{2\mMK}\,(\sqrt R+\sqrt N)+t_1, \\
\omega_2&:=\sigma\sqrt{2K}\,(\sqrt R+\sqrt N)+t_2, \\
\omega_3&:=\sigma\sqrt{2M}\,(\sqrt R+\sqrt N)+t_3, \\
\omega_4&:=\sigma\sqrt{(L^2-L)K}+t_4.
\end{align*}
Then:
\begin{enumerate}
\item {\bf Oracle gain.} Suppose $K\le RM-L$, $\omega_3 < \frac{1}{2} \gamma_3$, and $\omega_4 \le \gamma_3-\omega_3$. Fix $\varepsilon = 1/16$, and $\mu\in(0,\,1-L\mMK/RM)$. Define
\begin{align*}
\underline{a} & := 1-\left(\tfrac{\omega_3}{\gamma_3-\omega_3}\right)^2, \\
\overline{a} & := 1 - \left( \frac{\sigma\sqrt{[\,RM(1-\mu)-L \mMK\,]_+}} {\gamma_3+\omega_3} \right)^2, \\
\overline{E}&:=\Bigg(\frac{\omega_4}{\gamma_3-\omega_3}\Bigg)^{\!2}.
\end{align*}
If $\frac{4\sqrt{7}-2}{27} \le \underline{a} \le \overline{a} < 1$, then with probability at least $1-e^{-t_3^2/(\sigma^2\mMK)}-e^{-t_4^2/(\sigma^2\mMK)}-p_{\mathrm{w}}\!\big(\varepsilon,\mu\big)$, the oracle gain is lower bounded as follows:
\begin{equation}
\gor \ge 
\underbrace{\sqrt{1-\overline{a}} - \frac{\sqrt{\overline{E}}}{\sqrt{\overline{a}+\overline{E}}}}_{\underline{\gorlower}}.
\label{eq:finalgor}
\end{equation}
\item {\bf Empirical loss.} Suppose $\omega_1 < \gamma_1$, $\omega_2 < \gamma_2$, and $\omega_3 < \frac{1}{2} \gamma_3$. Define
\begin{align*}
\overline{\eta}&:=\frac{\omega_1}{\gamma_1-\omega_1}+\frac{\omega_2}{\gamma_2-\omega_2}, \\
\underline{\rho}&:=\sqrt{1-\Big(\tfrac{\omega_3}{\gamma_3-\omega_3}\Big)^2}.
\end{align*}
If $\overline{\eta} \le \frac{1}{2} \underline{\rho}$, then with probability at least $1-e^{-t_1^2/(\sigma^2\mMK)}-e^{-t_2^2/(\sigma^2\mMK)}-e^{-t_3^2/(\sigma^2\mMK)}$, the empirical loss is upper bounded as follows:
\begin{equation}
\lem \le \underbrace{\frac{\overline{\eta}}{\underline{\rho}-\overline{\eta}}}_{\overline{\lemupper}}.
\label{eq:finallem}
\end{equation}
\item {\bf Overall tensor gain.} Suppose all assumptions above are satisfied and that $\overline{\lemupper} < \underline{\gorlower}$. Then with probability at least $1-e^{-t_1^2/(\sigma^2\mMK)}-e^{-t_2^2/(\sigma^2\mMK)}-e^{-t_3^2/(\sigma^2\mMK)}-e^{-t_4^2/(\sigma^2\mMK)}-p_{\mathrm{w}}\!\big(\varepsilon,\mu\big)$, the overall tensor gain is lower bounded as follows:
\begin{equation}
\text{overall tensor gain}~\ge~ \underline{\gorlower} - \overline{\lemupper}. 
\label{eq:finalgain}
\end{equation}
\end{enumerate}
\end{theorem}

Theorem~\ref{thm:main} defines $\overline{\lemupper}$ (a probabilistic upper bound on $\lem$) and $\underline{\gorlower}$ (a probabilistic lower bound on $\gor$). Both terms depend on $\sigma$, $M$, $K$, $R$, $L$, and $\gamma_3$. However, only $\overline{\lemupper}$ depends on $\gamma_1$ and $\gamma_2$. Our numerical examples identify scenarios where $\overline{\lemupper} \le \underline{\gorlower}$.

\begin{table*}[h]
\centering
\caption{Parameter Settings for the Six Experimental Setups}
\label{tab:simulation_parameters}
\begin{tabular}{@{}lccccccc@{}}
\toprule
\textbf{Setup} & \textbf{N} & \textbf{M} & \textbf{P} & \boldmath{$L_p$} & \boldmath{$\omega$} & \boldmath{$\theta$} & \textbf{Description} \\ \midrule
($i$)   & 12 & 8 & 2 & [2, 3]    & [0.45, 0.5]      & [35$^\circ$, -15$^\circ$] & Two sources \\
($ii$)  & 10 & 6 & 2 & [2, 3]    & [0.45, 0.4]      & [35$^\circ$, 33$^\circ$]  & Two close sources \\
($iii$) & 12 & 6 & 3 & [1, 2, 3] & [0.45, 0.4, 0.35] & [-3$^\circ$, 35$^\circ$, 4$^\circ$] & Three sources; two close \\
($iv$)  & 12 & 8 & 2 & [2, 3]    & [0.5, 0.25]     & [35$^\circ$, -15$^\circ$] & Harmonically related \\ 
($v$)   & 13 & 8 & 2 & [2, 3]    & [0.45, 0.5]    &
[35$^\circ$, -15$^\circ$] & Coherent sources (i.e. $\alpha_{p, l} = 1$) \\ 
($vi$)  & 13 & 8 & 2 & [2, 3]  & [0.3, 0.315]  &
[35$^\circ$, -25$^\circ$] & Two sources, \textbf{$R=12$} in this case \\
\bottomrule
\end{tabular}
\end{table*}

\section{Numerical Examples} 
\label{sec:examples}

\subsection{Performance evaluation and comparison with matrix methods}

We simulate a uniform linear array consisting of \( R =15  \) microphone sensors with inter-element spacing \( d = c/f_s \), where the speed of sound is set to \( c = 340 \,\mathrm{m/s} \) and the sampling frequency is \( f_s = 8 \,\mathrm{kHz} \). We vary the number of samples \( N  \), the number of sources $P$, the data matrix dimension $M$, and the signal-to-noise ratio $\text{SNR} = 10 \log_{10} \frac{\|\mathcal{S}\|_\text{HS}^2}{\|\mathcal{Q}\|_\text{HS}^2}$, where $\|\cdot\|_{\text{HS}}$ is the Hilbert-Schmidt norm of a tensor. Our six setups are summarized in Table~\ref{tab:simulation_parameters}.

\begin{figure}[!t]
\includegraphics[width=0.99\columnwidth]{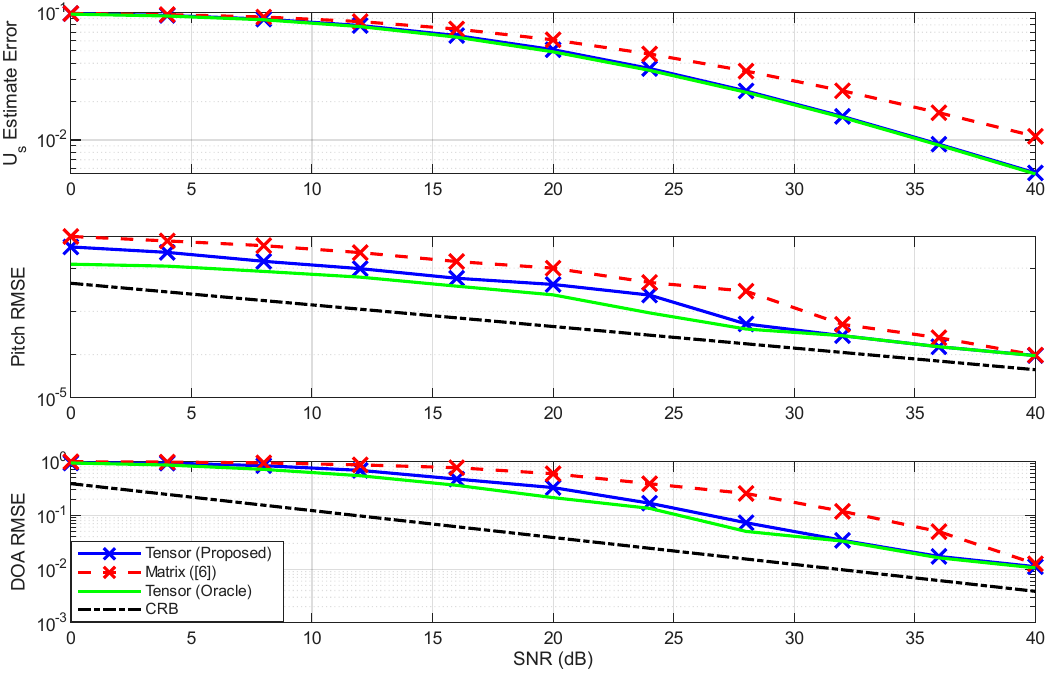} 
\caption{Comparison of the proposed tensor method with matrix method \cite{wu2014joint} for two sources (setup ($i$)).}
\label{rmse}
\end{figure}

\begin{figure}[!t]
\includegraphics[width=0.99\columnwidth]{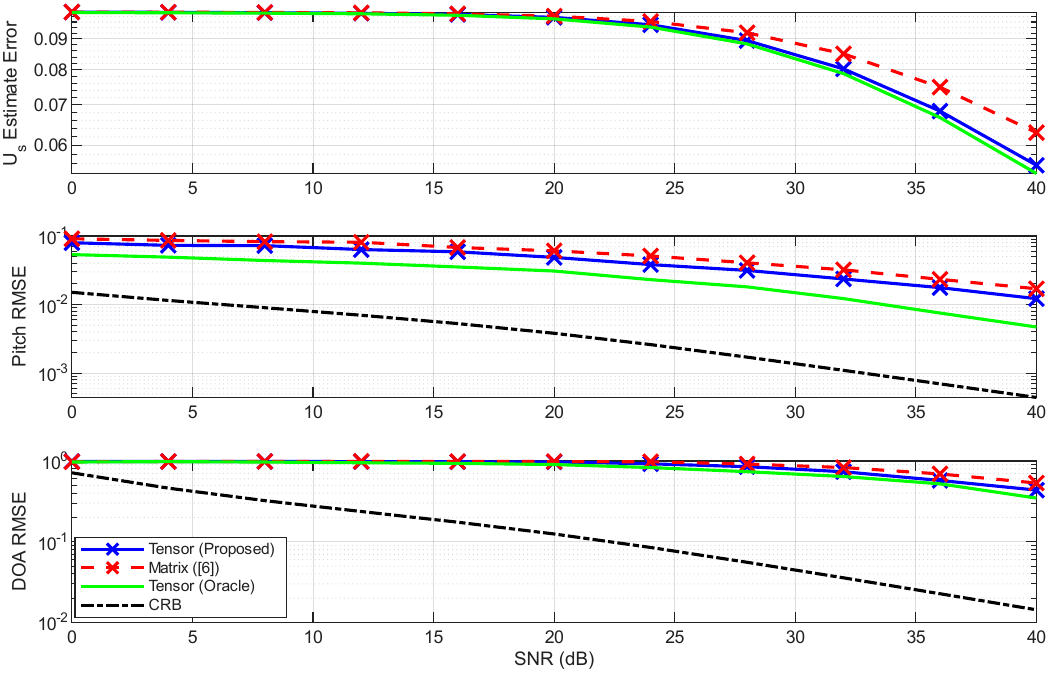} 
\caption{Comparison 
for two close sources (setup ($ii$)).}
\label{rmse2}
\end{figure}

\begin{figure}[!t]
\includegraphics[width=0.99\columnwidth]{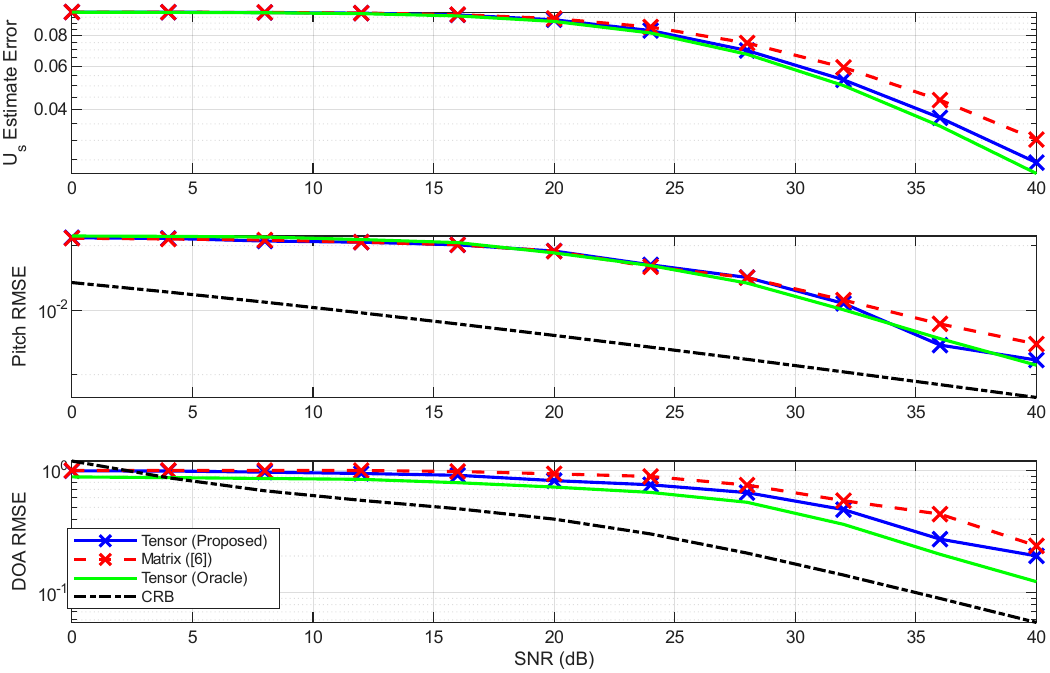}
\caption{Comparison 
for three sources with two close sources (setup ($iii$)).}
\label{rmse3}
\end{figure}

\begin{figure}[!t]
\includegraphics[width=0.99\columnwidth]{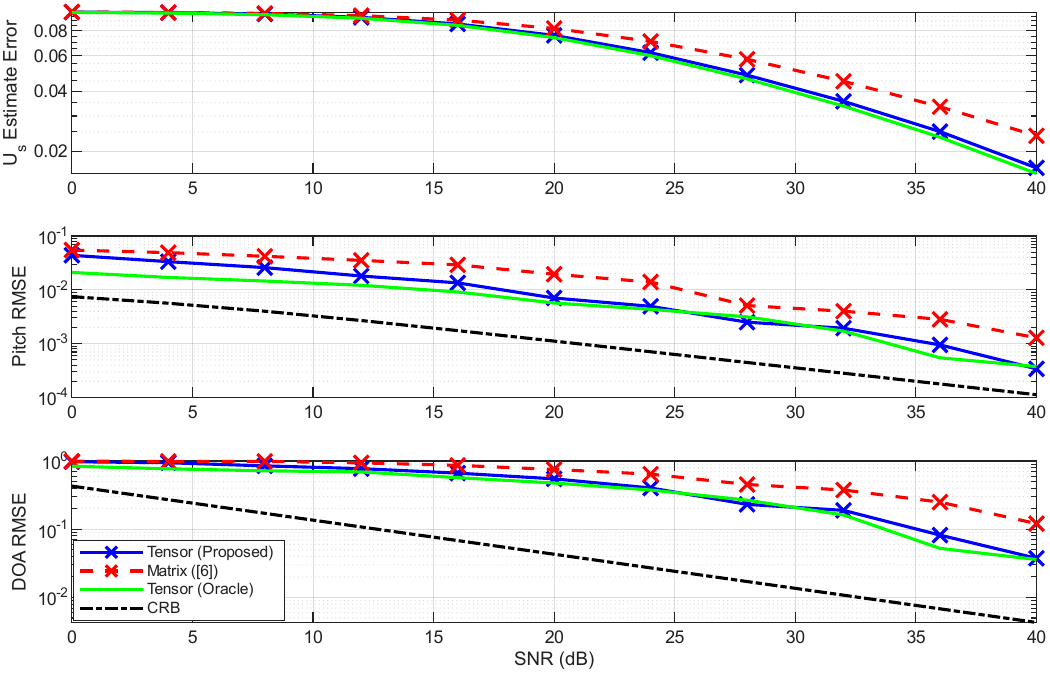} 
\caption{Comparison 
when there are two sources and the pitches of these two sources are harmonically related (setup ($iv$)).}
\label{rmse4}
\end{figure}

\begin{figure}[!t]
\includegraphics[width=0.99\columnwidth]{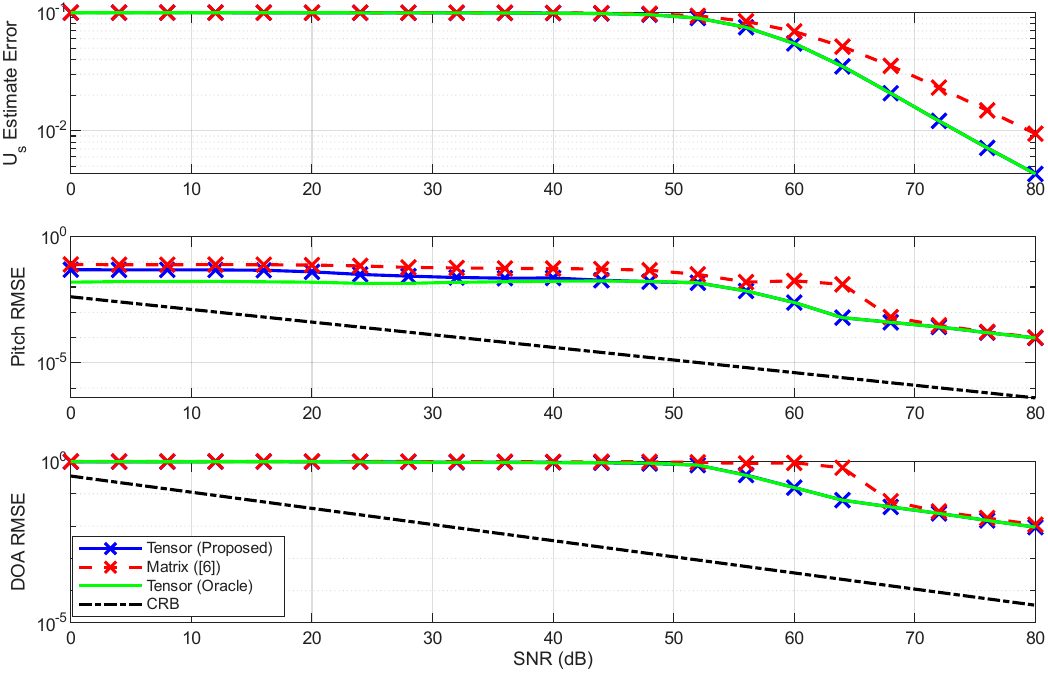} 
\caption{Comparison 
for two coherent sources (setup ($v$)).}
\label{rmse5}
\end{figure}
 
Figures \ref{rmse}--\ref{rmse5} show the corresponding numerical results. In each figure, the top plot reports the subspace estimation error, comparing $d(\Uten, \mathbf{U}_{\mathbf{S}_3})$ to $d(\Umat, \mathbf{U}_{\mathbf{S}_3})$; recall that~\cite[Sec.~III-D]{wu2014joint} is equivalent to using $\Umat$. We also report $d(\Uorc, \mathbf{U}_{\mathbf{S}_3})$, where $\Uorc$ is the ``oracle'' tensor-based subspace estimate discussed in Section~\ref{sec:proofsetup}. The middle and bottom plots report the root mean square error (RMSE) of the resulting parameter estimates. The RMSEs of the pitch and DOA estimation are defined as $\text{Pitch RMSE} = \sqrt{\frac{1}{P}\sum_{p = 1}^P (\widehat{\omega}_p - \omega_p)^2}$ and $\text{DOA RMSE} = \sqrt{\frac{1}{P}\sum_{p = 1}^P (\widehat{\theta}_p - \theta_p)^2}$. 


We see that the tensor method consistently outperforms the matrix method across the entire SNR range in subspace estimation, exhibiting a steeper decay in estimation error as the SNR increases. This improved subspace accuracy directly translates into enhanced parameter estimation performance. For pitch estimation, the tensor-based approach achieves lower RMSE at low to moderate SNRs, while both methods converge to similar performance levels at high SNRs, indicating asymptotic efficiency. In DOA estimation, the tensor method again demonstrates clear advantages, particularly in the low-SNR regime, where it yields significantly smaller RMSE than the matrix-based counterpart. 


While setups $(i)$--$(iv)$ corresponded to incoherent sources with randomly generated complex amplitudes $\alpha_{p, l}$, setup $(v)$ corresponds to coherent sources with all $\alpha_{p,l}=1$. In practical scenarios, such as indoor room acoustics, sound reflections frequently lead to such coherence \cite{kuttruff2024room}. In this case, while $\mathbf{Z}$ still has rank $L$ (recall Proposition~\ref{prop:rank}), it becomes poorly conditioned compared to the incoherent case. As shown in Figure~\ref{rmse5}, this introduces significant challenges for both estimation techniques. Nevertheless, the proposed tensor method consistently outperforms the matrix-based approach. 

\subsection{Performance analysis and comparison with state-of-the-art tensor methods}
We benchmark the proposed tensor method against the matrix baseline \cite{wu2014joint}, CP-ALS \cite{sidiropoulos2017tensor}, and TT-MDL \cite{gong2020tensor}, one of the state-of-the-art tensor methods (see Figure \ref{rmse6}). The proposed method attains the lowest subspace estimation error against the matrix baseline and TT-MDL at every tested SNR, but its ranking against CP-ALS is SNR-dependent: CP-ALS attains lower subspace error in a mid-range band from 8 to 20 dB, while the proposed method is strictly better at the low end (0--4 dB) and, more markedly, from 24 dB upward, where the performance of CP-ALS saturates.

The CP model represents the tensor as a sum of $L$ rank-one components, with each component coupling a spatial steering vector, a temporal steering vector, and a snapshot coefficient vector across the three modes. In contrast, the Tucker-type relaxation underlying the proposed method only requires the individual mode unfoldings to have rank at most $L$, without enforcing a one-to-one correspondence among their latent components. Consequently, CP exploits a stronger structural model and can provide higher estimation accuracy when the factors are well conditioned. Its performance may nevertheless be sensitive to initialization and to near-collinearity among the factor-matrix columns. Although near-collinearity does not necessarily violate Kruskal’s sufficient uniqueness condition, it can make the CP decomposition poorly conditioned and numerically unstable. The proposed method sacrifices some of this stronger component-level coupling in exchange for a simpler noniterative SVD-based implementation, greater robustness to ill-conditioning, and explicit non-asymptotic performance guarantees.

\begin{figure}[!t]
\includegraphics[width=0.99\columnwidth]{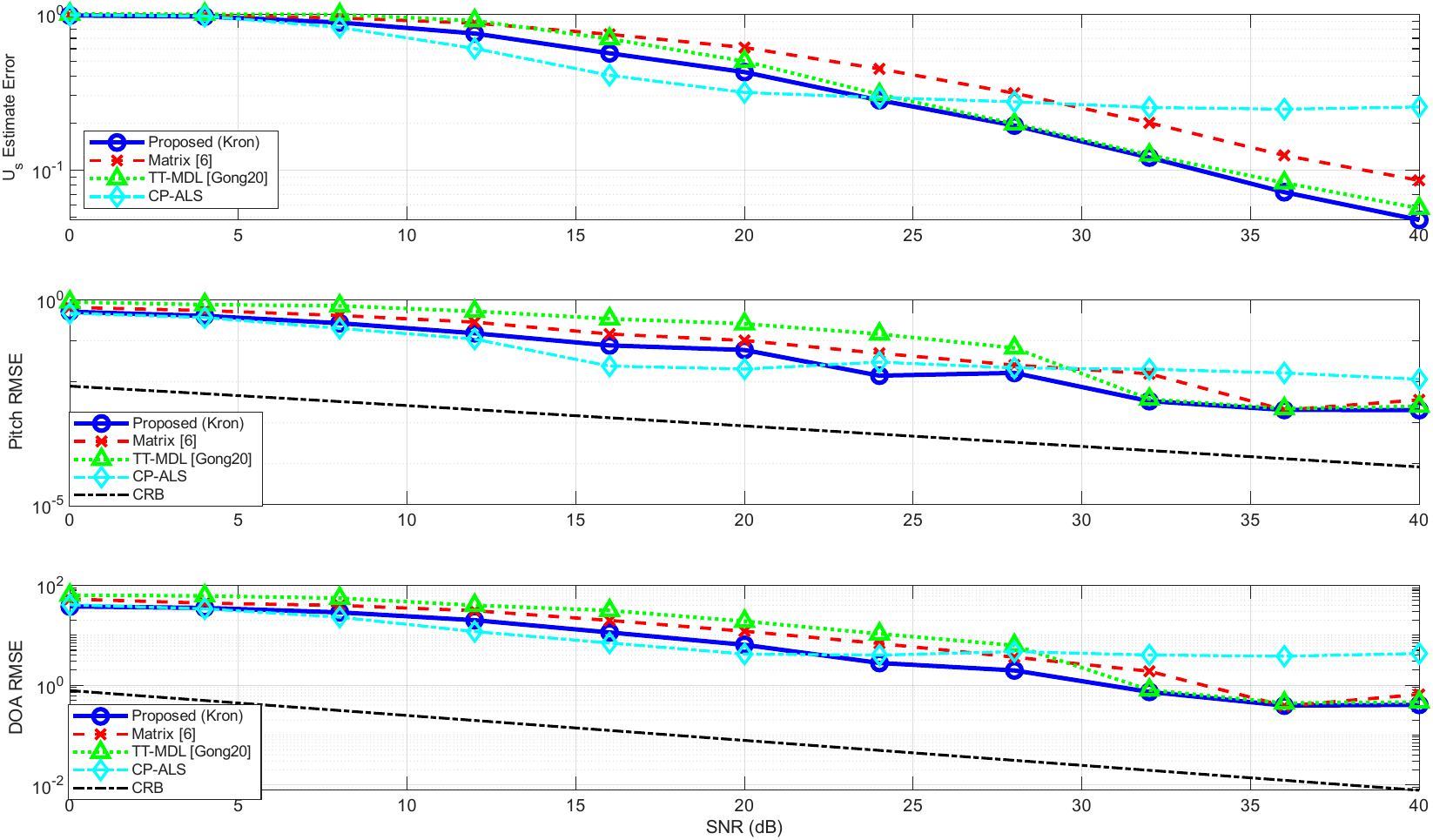} 
\caption{Comparison of the proposed tensor method with matrix method \cite{wu2014joint}, CP-ALS \cite{sidiropoulos2017tensor} and TT-MDL method \cite{gong2020tensor} (setup ($vi$)). Note in this setup, $R = 12$. }
\label{rmse6}
\end{figure}

\subsection{Tensor gain bound characterization}
Figure~\ref{bound_demo1} illustrates the deterministic bounds on the oracle gain and empirical loss derived in Sections~\ref{sec:oraclegain} and~\ref{sec:emploss}. The array configuration for this experiment is described in the figure caption. We vary the noise standard deviation and compare the realized oracle gain and empirical loss with their corresponding theoretical bounds in panels (a) and (b), respectively. Because the oracle gain is lower bounded, all points in panel (a) lie on or below the diagonal line, which represents equality between the measured quantity and its bound. Similarly, because the empirical loss is upper bounded, all points in panel (b) lie on or above the diagonal line. These results numerically demonstrate the validity of the bounds established in Theorems~\ref{thm:oraclebounds} and~\ref{thm:exact}.



Figure~\ref{sweet_demo1} illustrates the ``sweet-spot'' behavior of the empirical tensor gain. The gain is small when the matrix-based estimate is already highly accurate, typically at high SNR, increases as the matrix estimate enters an intermediate-quality regime, and decreases again when the matrix estimate becomes too inaccurate, typically at low SNR. At high SNR, little room remains for further refinement. At intermediate SNR, the matrix estimate contains a substantial component outside the Kronecker-structured signal cage that can be removed by the tensor refinement, leading to the largest gain (marked as star in Figure~\ref{sweet_demo1}). At very low SNR, however, the empirically estimated Kronecker projector becomes unreliable, and the benefit of the refinement diminishes.


Figure \ref{sweet_demo3} provides a representative illustration of the impact of the Kronecker projector error $\eta$ on the range of $a$ where positive overall tensor gain is guaranteed according to our theory. We refer to this as the ``certified interval'' of $a$. Larger $\eta$ indicates the empirical Kronecker projector is less accurate. As $\eta$ increases, the certified interval narrows and eventually disappears. This behavior reinforces that the tensor-gain ``sweet spot'' represents a finite operating region that depends critically on the estimation accuracy of the Kronecker projector.


Finally, Figure \ref{prob_demo} presents two practical scenarios in which $\overline{\lemupper}$, the probabilistic upper bound on the empirical loss, falls below $\underline{\gorlower}$, the lower bound on the oracle gain. In such scenarios, our Theorem \ref{thm:main} guarantees positive overall tensor gain with high probability. In this figure, we define the certified SNR as the smallest SNR at which $\overline{\lemupper} \le \underline{\gorlower}$, i.e., where the dashed green curve exceeds the solid red curve. These examples demonstrate that the probabilistic guarantee in Theorem \ref{thm:main} can be practically achieved. Nevertheless, the guarantee remains conservative, as certification requires relatively favorable operating conditions. In practice, positive empirical tensor gain is often observed under substantially less demanding conditions.


\begin{figure}[!t]
\includegraphics[width=0.99\columnwidth]{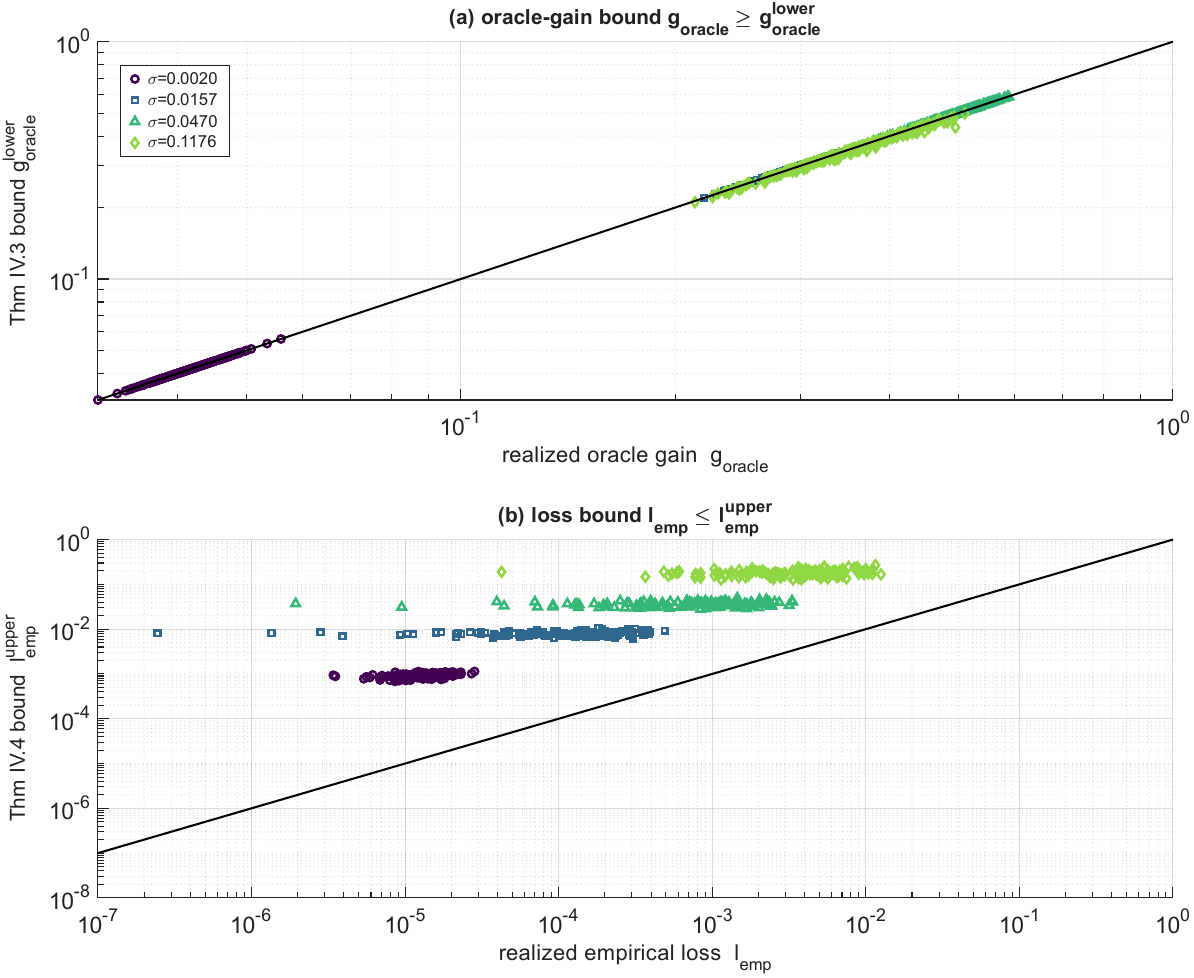} 
\caption{Demonstration of the deterministic bounds: (a) oracle gain lower bound, (b) empirical loss upper bound. $R=15$, $M=8$, $N=12$, $K=5$, $P = 2$, $L_p = [2, 3]$, $\omega_p = [1.3, 1.0]$, $\theta_p = [50^\circ, -35^\circ]$. The sources are coherent (i.e., $\alpha_{p, l} = 1$). For each of four fixed noise levels $\sigma$, $200$ trials are run.}
\label{bound_demo1}
\end{figure}

\begin{figure}[!t]
\includegraphics[width=0.99\columnwidth]{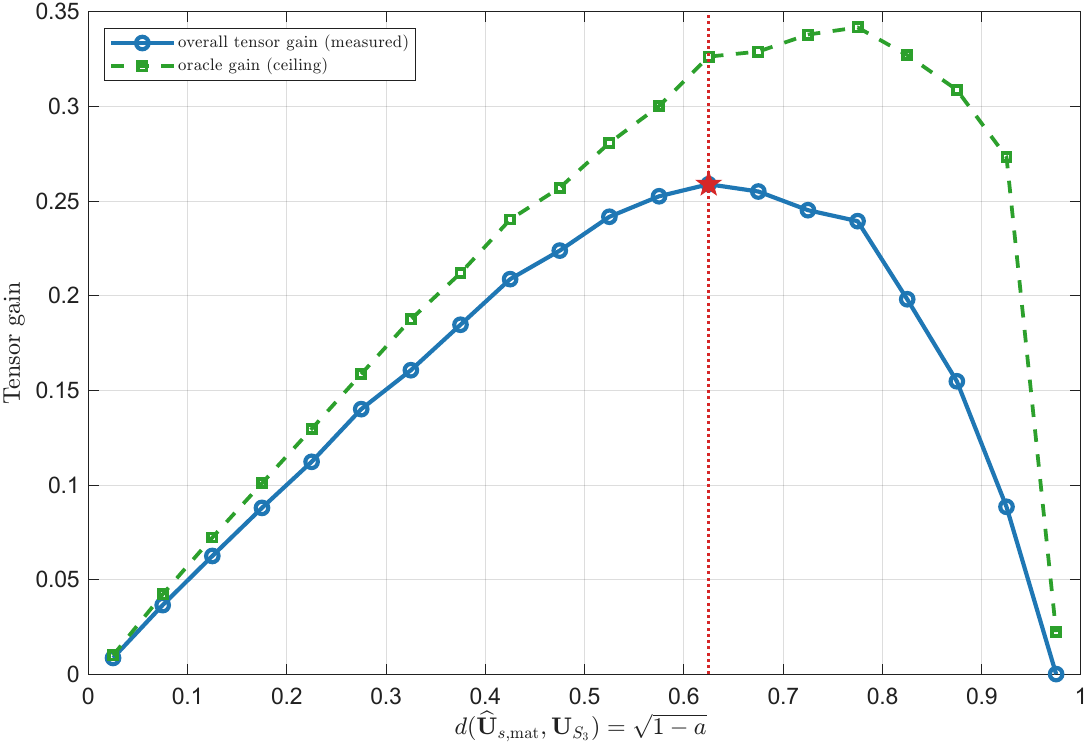} 
\caption{Tensor gain versus the matrix-baseline error. $R = 20$, $M = 10$, $N = 24$, $P= 2$, $L_p = [2, 3]$, $\omega_p =[0.45, 0.55]$, $\theta_p = [35^\circ, -20^\circ]$. SNR is swept from $0$ to $60$ dB with $2$ dB step size and with $300$ trials for each SNR. The amplitudes of sources are complex random.
Every trial produces a unique value for both $d(\Umat,\Utrue)$ and tensor gain; these trials are grouped into bins according to $d(\Umat,\Utrue)$ and the median tensor gain for each bin is plotted.
}
\label{sweet_demo1}
\end{figure}


\begin{figure}[!t]
\includegraphics[width=0.99\columnwidth]{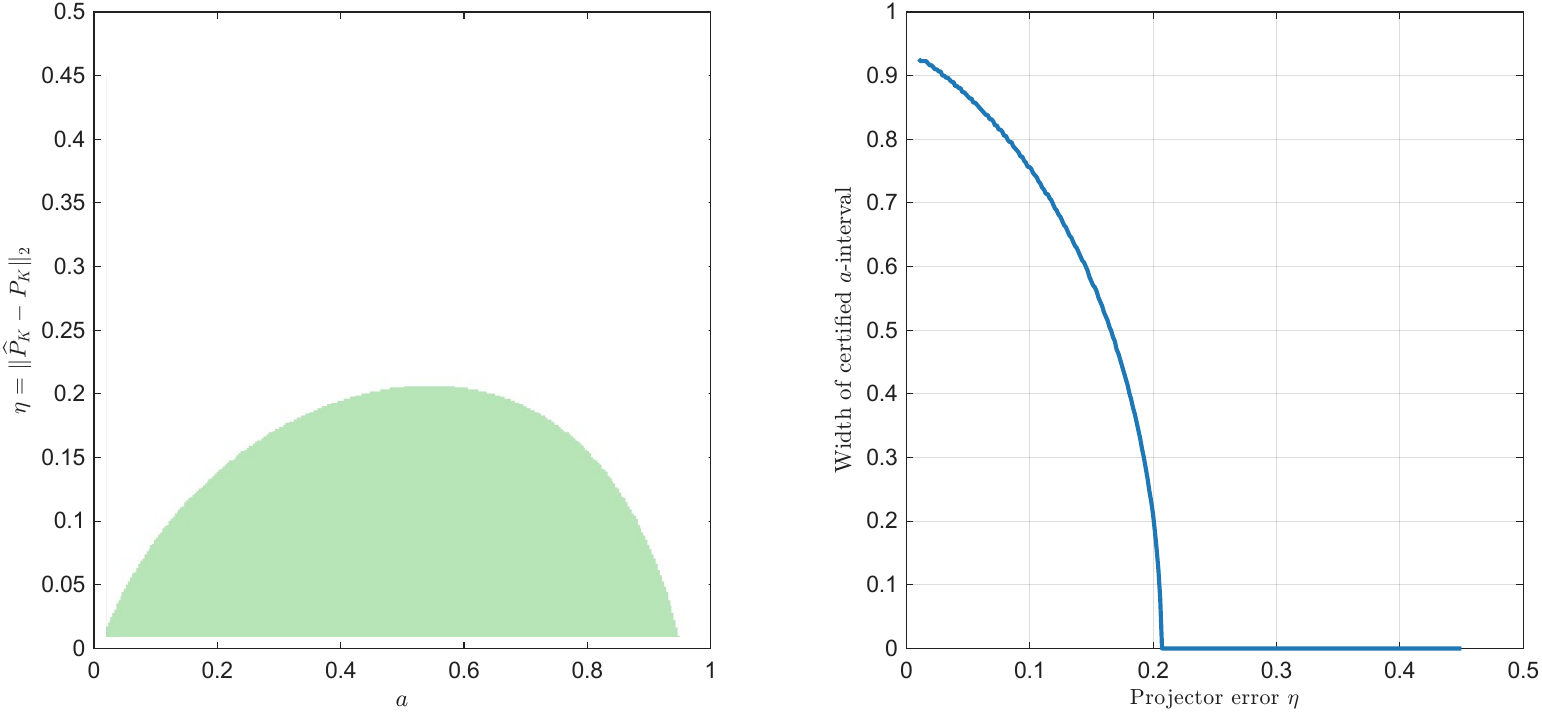} 
\caption{The certified interval versus the projector error $\eta$. 
We set $\|\Epar\|_2 = 0.22$ as a representative fixed value. The left panel marks the certified region in the $(a,\eta)$ plane. At every $(a, \eta)$ grid point, we ``certify'' (shade) the point if $\gorlower > \frac{\eta}{\sqrt{a} - \eta}$, noting that $\frac{\eta}{\sqrt{a} - \eta} \ge \lemupper$ due to Theorem \ref{thm:exact} and Lemma \ref{lem:rholower}.}
\label{sweet_demo3}
\end{figure}

\begin{figure}[!t]
\includegraphics[width=0.99\columnwidth]{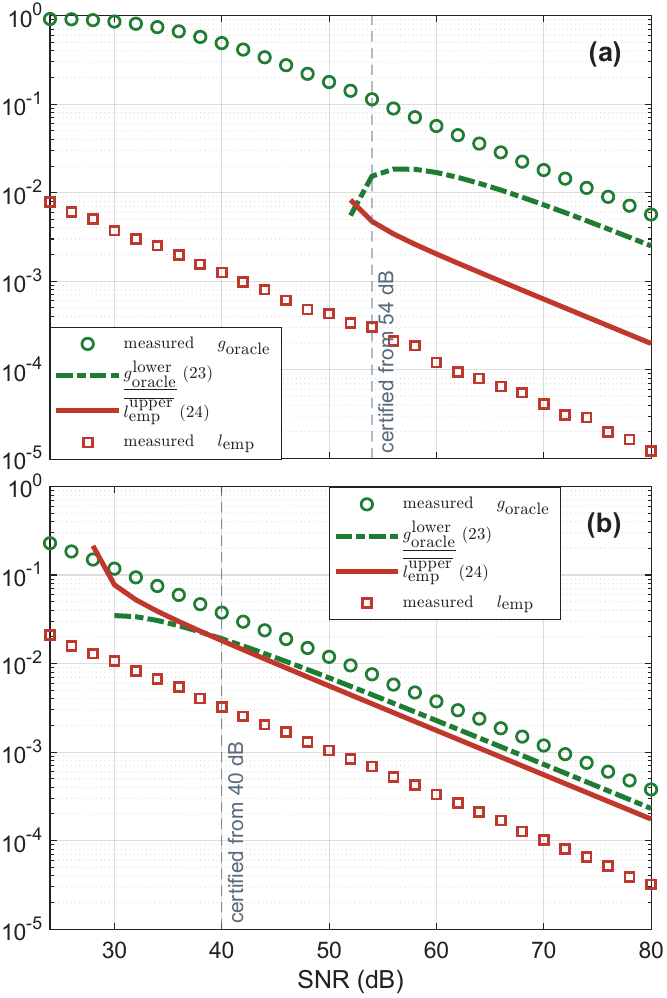} 
\caption{Demonstration of two scenarios where $\overline{\lemupper} \le \underline{\gorlower}$: (a)~$R = M = 60$, $N = 64$, $K = 5$, $P = 2$, $L_p = [2,3]$, $L = 5$, $\omega_p = [0.30, 0.95]$, $\theta_p = [65^\circ, -65^\circ]$; (b) $R = M = 30$, $N = 31$, $K = 2$, $P = 2$, $L_p = [1,1]$, $L = 2$, $\omega_p = [0.25, 0.95]$, $\theta_p = [55^\circ, -40^\circ]$. The amplitudes of sources for both scenarios are complex random. All slack variables $t_1,t_2,t_3,t_4$ are set such that the corresponding upper bound on failure probability $\exp(-t_i^2/(\sigma^2\mMK))=0.02$. We budget another $0.02$ failure probability by restricting $p_{\mathrm{w}}(\varepsilon,\mu) \le 0.02$, setting $\varepsilon = 1/16$, and solving for the smallest allowable $\mu$. By the union bound, every plotted guarantee then holds with probability at least $1-5(0.02)=0.90$.  Each measured point is obtained via taking the smallest measured $\gor$ (or largest measured $\ell_{\textrm{emp}}$) over $300$ trials for each SNR.}
\label{prob_demo}
\end{figure}

\section{Conclusions}
\label{sec:conclusion}

This paper developed a tensor-based framework for joint pitch and DOA estimation of multiple harmonic sources. By arranging the received array data as a multidimensional Hankel tensor, we showed that the true mode-3 signal subspace lies in a Kronecker-structured signal cage determined by the mode-1 and mode-2 subspaces. This observation led to a tensor-refined subspace estimator that projects the conventional matrix-based estimate onto an empirically estimated Kronecker cage before applying ESPRIT-type pitch and DOA retrieval.

The paper also provided a theoretical explanation for the tensor gain. We proved that the oracle refinement using the true Kronecker projector cannot degrade the matrix estimate when the matrix estimate has nonzero overlap with the true subspace. We then decomposed the empirical gain into oracle improvement and empirical projector loss, yielding a deterministic sufficient condition for positive tensor gain. Under additive complex Gaussian noise, this condition was further converted into a high-probability guarantee using subspace perturbation bounds and concentration results for the structured Hankel noise unfoldings.

The numerical results support both the estimation method and the theory. The proposed tensor refinement consistently improves subspace accuracy and pitch/DOA RMSE in low-to-moderate SNR regimes, including closely spaced sources, harmonically related sources, and coherent-source settings. Additional experiments verify the deterministic bounds, demonstrate the sweet-spot phenomenon in which gain is largest at intermediate matrix-estimate quality, and show that the certified gain window shrinks as the empirical Kronecker projector becomes less accurate. Comparisons with the matrix baseline and TT-MDL further indicate that the proposed Kronecker-refined estimator provides a favorable balance between structure exploitation and subspace-estimation robustness. Our theoretical analysis decomposes the tensor gain into multiple factors that one could attempt to optimize through the array configuration, presenting an interesting opportunity for possible future work.

\section*{Acknowledgments}

This work was partially supported by NSF Grant CCF-2106834. 

The authors acknowledge the use of Claude, ChatGPT and Gemini for analytical and grammatical feedback, as well as collaborative code and text generation impacting all sections of the paper. AI-generated content has been carefully verified and heavily edited (or entirely rewritten) by the authors, who take full responsibility for the final work.

\bibliographystyle{IEEEtran}
\bibliography{refs}

\newpage

\appendix[(Supplementary Material)]

\subsection{Supporting proofs for Section~\ref{sec:oraclegain}}

We first prove Lemma~\ref{lem:floor}.

\begin{proof}
By~\eqref{eq:fixed}, $\range(\Utrue)\subseteq\range(\PK)$, hence
$\Utrue\Utrue^{H}\preceq\PK$ as orthogonal projectors and
$\spec{\PK\mathbf v}\ge\spec{\Utrue\Utrue^{H}\mathbf v}$ for all
$\mathbf v$. Therefore 
\begin{align}
\rho &= \sigmaL(\PK\Umat) = \min_{\spec{\mathbf w}=1}\spec{\PK\Umat\mathbf w} \nonumber \\
& \ge\min_{\spec{\mathbf w}=1}\spec{\Utrue^{H}\Umat\mathbf w} 
= \sigmaL(\Utrue^{H}\Umat), \label{eq:rhovsa}
\end{align}
since $\Utrue$ has orthonormal columns. The singular values of the
$L\times L$ matrix $\Utrue^{H}\Umat$ are the cosines of the principal
angles, so $\sigmaL(\Utrue^{H}\Umat)=\cos\theta_{\max}=\sqrt{1-d^2(\Umat,\Utrue)}$,
which is positive (and the rank is $L$) when $d(\Umat,\Utrue)<1$.
\end{proof}

We now complete the proof of Proposition~\ref{prop:oracle}.

\begin{proof}
If $d(\Umat,\Utrue) = 1$, the proposition follows immediately from the fact that $d(\Uorc,\Utrue) \le 1$.

If $d(\Umat,\Utrue) < 1$, let $\PK\Umat=\mathbf Q\mathbf R$ be a thin QR factorization; $\mathbf R$
is invertible by Lemma~\ref{lem:floor}. The principal-angle cosines
between $\range(\PK\Umat)$ and $\range(\Utrue)$ are the singular values of
$\Utrue^{H}\mathbf Q=(\Utrue^{H}\Umat)\mathbf R^{-1}$, where we used~\eqref{eq:fixed}. Since
$\mathbf R^{H}\mathbf R=\Umat^{H}\PK\Umat\preceq\Umat^{H}\Umat=\mathbf I$,
we have $\spec{\mathbf R}\le1$ and $\sigmaL(\mathbf R^{-1})\ge1$, so
\begin{align*}
\cos\theta_{\max}^{\mathrm{or}}
&:= \sigmaL\!\big((\Utrue^{H}\Umat)\mathbf R^{-1}\big) \\
&\ge\sigmaL(\Utrue^{H}\Umat)=:\cos\theta_{\max}^{\mathrm{mat}}.
\end{align*}
Hence $\sin\theta_{\max}^{\mathrm{or}}\le\sin\theta_{\max}^{\mathrm{mat}}$,
i.e.\ $d(\Uorc,\Utrue)\le d(\Umat,\Utrue)$. This completes the proof of Proposition ~\ref{prop:oracle}.
\end{proof}

From the derivation above, we see that
\begin{align}
\gor &:= d(\Umat,\Utrue)-d(\Uorc,\Utrue) \nonumber \\
&= \sin\theta_{\max}^{\mathrm{mat}}-\sin\theta_{\max}^{\mathrm{or}} \nonumber \\
&= \sqrt{1-\cos^2\theta_{\max}^{\mathrm{mat}}} - \sqrt{1-\cos^2\theta_{\max}^{\mathrm{or}}} \nonumber \\
&= \sqrt{1-\sigmaL^2(\Utrue^{H}\Umat)} \nonumber \\ & \qquad - \sqrt{1-\sigmaL^2\!\big((\Utrue^{H}\Umat)\mathbf R^{-1}\big)} \nonumber \\
&:= \sqrt{1-a} - \sqrt{1-b}, \label{eq:glowerab}
\end{align}
where $a:=\sigmaL^2(\Utrue^{\herm}\Uhat), b:=\sigmaL^2\!\big(\Utrue^{\herm}\Uhat\mathbf{R}^{-1}\big)$. We establish a lower bound on $b$.

\begin{lemma}[In-cage retained energy]\label{lem:incageenergy}
We have
\begin{equation}\label{eq:bL}
b\;\ge\;\blower:=\frac{a}{a+\|\Epar\|_2^2}.
\end{equation}
\end{lemma}

\begin{proof}
Because $\PK$ is an orthogonal projector,
$\mathbf{R}^{\herm}\mathbf{R}=\Uhat^{\herm}\PK^{\herm}\PK\Uhat=\Uhat^{\herm}\PK\Uhat$.
Split $\PK=\PiS+(\PK-\PiS)$ into orthogonal projectors with orthogonal ranges (as
$\range(\Utrue)\subseteq\range(\PK)$). Writing $\mathbf{M}:=\Utrue^{\herm}\Uhat$ and using the definition of $\Epar$,
\begin{align*}
&\mathbf{R}^{\herm}\mathbf{R} \\&=\Uhat^{\herm}\PiS\,\Uhat+\Uhat^{\herm}(\PK-\PiS)\Uhat \\
&=\mathbf{M}^{\herm}\mathbf{M}+\Epar^{\herm}\Epar,
\end{align*}
since $(\PK-\PiS)$ is idempotent and Hermitian. With $\mathbf{z}:=\mathbf{R}^{-1}\mathbf{x}$,
\begin{align*}
b&=\sigmaL^2(\mathbf{M}\mathbf{R}^{-1})
=\min_{\mathbf{z}\neq0}
\frac{\mathbf{z}^{\herm}\mathbf{M}^{\herm}\mathbf{M}\,\mathbf{z}}
     {\mathbf{z}^{\herm}\mathbf{R}^{\herm}\mathbf{R}\,\mathbf{z}} \\
&=\min_{\mathbf{z}\neq0}
\frac{\mathbf{z}^{\herm}\mathbf{G}_1\mathbf{z}}
     {\mathbf{z}^{\herm}\mathbf{G}_1\mathbf{z}+\mathbf{z}^{\herm}\mathbf{G}_2\mathbf{z}} =\frac{1}{1+\lmax(\mathbf{G}_1^{-1}\mathbf{G}_2)},
\end{align*}
where $\mathbf{G}_1:=\mathbf{M}^{\herm}\mathbf{M}\succ0$ (as $a=\lmin(\mathbf{G}_1)>0$ due to Assumption 5)
and $\mathbf{G}_2:=\Epar^{\herm}\Epar\succeq0$. Bounding the generalized eigenvalue,
\begin{equation*}
\lmax(\mathbf{G}_1^{-1}\mathbf{G}_2)
=\max_{\mathbf{z}\neq0}
\frac{\mathbf{z}^{\herm}\mathbf{G}_2\mathbf{z}}{\mathbf{z}^{\herm}\mathbf{G}_1\mathbf{z}}
\le\frac{\lmax(\mathbf{G}_2)}{\lmin(\mathbf{G}_1)}=\frac{\|\Epar\|_2^2}{a},
\end{equation*}
so $b\ge\big(1+\|\Epar\|_2^2/a\big)^{-1}=a/(a+\|\Epar\|_2^2)=\blower$.
\end{proof}

We are now prepared to prove Theorem~\ref{thm:oraclebounds}.

\begin{proof}
From~\eqref{eq:glowerab},
$g_{\mathrm{oracle}}=\sqrt{1-a}-\sqrt{1-b}$ with
$b=\sigmaL^2(\Utrue^{\herm}\Uhat\mathbf{R}^{-1})$, and from Lemma~\ref{lem:incageenergy}, $b\ge \blower=a/(a+\|\Epar\|_2^2)$.

The map $t\mapsto\sqrt{1-t}$ is decreasing on $[0,1]$ and
$\blower\le b\le1$, so $\sqrt{1-b}\le\sqrt{1-\blower}$ and therefore
\begin{align*}
g_{\mathrm{oracle}}=\sqrt{1-a}-\sqrt{1-b} &\ge \sqrt{1-a}-\sqrt{1-\blower} \\
&= \sqrt{1-a} - \dfrac{\sqrt{\|\Epar\|_2^2}}{\sqrt{a+\|\Epar\|_2^2}}.
\end{align*}

Note also that 
\begin{align*}
\sigmaL^2\!\big((\Utrue^{H}\Umat)\mathbf R^{-1}\big)
&\leq \sigmaL^2\!\big(\Utrue^{H}\Umat\big) \|\mathbf{R}^{-1}\|_2^2 \\
&= \frac{\sigmaL^2\!\big(\Utrue^{H}\Umat\big)}{\rho^2}. 
\end{align*}
Therefore, we also have an upper bound on the oracle gain:
\[
\gor = \sqrt{1-a}-\sqrt{1-b} \le \sqrt{1-a}-\sqrt{1-a/\rho^2}.
\]
This completes the proof of Theorem~\ref{thm:oraclebounds}.
\end{proof}

\subsection{Supporting proofs for Section~\ref{sec:emploss}}

Here we prove Theorem~\ref{thm:exact}.

\begin{proof}
Define $\mathbf B:=\PK\Umat$ (so $\range(\mathbf B)=\range(\Uorc)$ and $\sigmaL(\mathbf B)=\rho$). Next, define 
$\widehat{\mathbf B}:=\PKhat\Umat$. We can then write $\widehat{\mathbf B} = \mathbf B + \mathbf F$, where $\mathbf F:=(\PKhat-\PK)\Umat$. It follows that $\spec{\mathbf F}\le\spec{\PKhat-\PK}\spec{\Umat}=\eta$. 

Since $\range(\widehat{\mathbf B})$ is $L$-dimensional under Assumption~5, Weyl's inequality gives $\sigmaL(\widehat{\mathbf B}) = \sigmaL(\mathbf{B+F})\ge \sigmaL(\mathbf{B}) - \|\mathbf{F}\|_2 \ge \rho-\eta > 0$. 

For any vector $\mathbf x$, $(\mathbf I-\Proj_{\mathbf B})\mathbf B\mathbf x=\mathbf 0$, hence $\spec{(\mathbf I-\Proj_{\mathbf B})\widehat{\mathbf B}\mathbf x} =\spec{(\mathbf I-\Proj_{\mathbf B})\mathbf F\mathbf x}\le\eta\spec{\mathbf x}$, and hence for any unit vector $\mathbf v=\widehat{\mathbf B}\mathbf x/\spec{\widehat{\mathbf B}\mathbf x}\in\range(\widehat{\mathbf B})$,
\begin{equation*}
\spec{(\mathbf I-\Proj_{\mathbf B})\mathbf v}
= \frac{\spec{(\mathbf I-\Proj_{\mathbf B})\widehat{\mathbf B}\mathbf x}}{\spec{\widehat{\mathbf B}\mathbf x}}  \le\frac{\eta}{\sigmaL(\widehat{\mathbf B})} \le\frac{\eta}{\rho-\eta}.
\end{equation*}
Taking the supremum over such $\mathbf v$ gives
$d(\Uten,\Uorc)\le\eta/(\rho-\eta)$. By the triangle inequality,
\begin{align*}
\lem &= d(\Uten,\Utrue)-d(\Uorc,\Utrue) \\
&\le (d(\Uorc,\Utrue) + d(\Uten,\Uorc)) - d(\Uorc,\Utrue) \\
&= d(\Uten,\Uorc) \le \frac{\eta}{\rho-\eta}.
\end{align*}
\end{proof}

\subsection{Supporting proofs for Section~\ref{sec:prob}}

The proof of Lemma~\ref{lem:rholower} (the fact that $\rho^2 \ge a$) is established in~\eqref{eq:rhovsa}. 

Next we prove Lemma~\ref{lem:kronecker_tele}.

\begin{proof}
We express the difference between the two Kronecker product projectors using a standard telescoping identity:
\begin{align*}
\PKhat - \PK &= \widehat{\mathbf T}_2 \otimes \widehat{\mathbf T}_1 - \mathbf T_2 \otimes \mathbf T_1 \\
&= (\widehat{\mathbf T}_2 - \mathbf T_2) \otimes \widehat{\mathbf T}_1 + \mathbf T_2 \otimes (\widehat{\mathbf T}_1 - \mathbf T_1).
\end{align*}
Noting that the spectral norm of a Kronecker product satisfies $\spec{\mathbf{A} \otimes \mathbf{B}} = \spec{\mathbf{A}}\spec{\mathbf{B}}$, we have:
\[
\spec{\PKhat - \PK} \le \spec{\widehat{\mathbf T}_2 - \mathbf T_2} \spec{\widehat{\mathbf T}_1} + \spec{\mathbf T_2} \spec{\widehat{\mathbf T}_1 - \mathbf T_1}.
\]
Note that $\spec{\widehat{\mathbf T}_1} = 1$ and $\spec{\mathbf T_2} = 1$. Substituting these values along with the definitions of $\delta_1$ and $\delta_2$ yields:
\[
\eta \le \delta_2 \cdot 1 + 1 \cdot \delta_1 = \delta_1 + \delta_2.
\]
\end{proof}

From this follows the proof of Corollary~\ref{cor:wedin_sin}.

\begin{proof}
By Wedin's $\sin\theta$ theorem \cite[Ch. V, Sec. 4.1, Theorem 4.4]{StewartSun1990}, \cite[Sec. 2.3]{CaiZhang2018} applied to the singular subspaces of a matrix under additive deterministic perturbations, the canonical distance between the true and estimated singular subspaces is bounded by:
\begin{align*}
\delta_i &= \spec{\widehat{\mathbf T}_i - \mathbf T_i} = \|\sin\Theta( \mathbf{\widehat U}_i,\mathbf{U}_i)\|_2  \\ 
&\le \frac{ \max\{\|\mathbf{Q}_i \mathbf{\widehat V}_i\|_2,\|\mathbf{\widehat U}_i^H \mathbf{Q}_i\|_2\} }{ \sigmaL(\mathbf{X}_i)} \le \frac{\spec{\mathbf{Q}_i}}{\sigmaL(\mathbf{X}_i)}                                \\
&\le \frac{\spec{\mathbf{Q}_i}}{\sigmaL(\mathbf{S}_i) - \spec{\mathbf{Q}_i}} = \frac{\spec{\mathbf{Q}_i}}{\gamma_i - \spec{\mathbf{Q}_i}},
\end{align*}
where the last inequality follows by Weyl's inequality: $\sigmaL(\mathbf{X}_i) = \sigmaL(\mathbf{S}_i + \mathbf{Q}_i) \ge \sigmaL(\mathbf{S}_i) - \|\mathbf{Q}_i \|_2$.

Combining this with Lemma~\ref{lem:kronecker_tele} yields the stated upper bound.
\end{proof}

Next we prove Lemma~\ref{lem:mode1noise}.

\begin{proof}
Let \(\mathcal{\mathbf H}_1:\mathbb C^{R\times N}\to\mathbb C^{R\times MK}\) 
denote the linear map that sends \(\mathbf Q\) to its mode-1 unfolding \(\mathbf Q_1\). For any \(\mathbf A,\mathbf B\in\mathbb C^{R\times N}\), the reverse triangle inequality and the inequality \(\|\cdot\|_2\le \|\cdot\|_F\) give 
\[ 
\left| \|\mathcal{\mathbf H}_1(\mathbf A)\|_2 - \|\mathcal{\mathbf H}_1(\mathbf B)\|_2 \right| \le \|\mathcal{\mathbf H}_1(\mathbf A-\mathbf B)\|_2 \le \|\mathcal{\mathbf H}_1(\mathbf A-\mathbf B)\|_F. 
\] 
Let \(\mathbf E=\mathbf A-\mathbf B\). Then 
\begin{align*} 
\|\mathcal{\mathbf H}_1(\mathbf E)\|_F^2 \le \min(M,K)\|\mathbf E\|_F^2,
\end{align*}
owing to the fact that $\mathcal{\mathbf H}_1(\mathbf E)$ contains repeated copies of the entries of $\mathbf E$, with no entry of $\mathbf E$ appearing more than $\min(M,K)$ times in $\mathcal{\mathbf H}_1(\mathbf E)$. Thus,
\begin{align*}
\left|
\| \mathbf{H}_1(\mathbf A)\|_2
-
\|\mathbf H_1(\mathbf B)\|_2
\right|
\le \sqrt{\min(M,K)}\|\mathbf A-\mathbf B\|_F.
\end{align*}
Hence the mapping \(\mathbf Q\mapsto \|\mathcal{\mathbf H}_1(\mathbf Q)\|_2\) is \(\sqrt{\min(M,K)}\)-Lipschitz. 

Now write
\[
    \mathbf Q=\sigma\mathbf W,
    \qquad
    \mathbf W=\frac{\mathbf G_R+i\mathbf G_I}{\sqrt2},
\]
where \(\mathbf G_R,\mathbf G_I\in\mathbb R^{R\times N}\) have i.i.d. standard
real Gaussian entries. Define
\[
    F(\mathbf G_R,\mathbf G_I)
    :=
    \left\|
    \mathbf H_1
    \left(
    \sigma\frac{\mathbf G_R+i\mathbf G_I}{\sqrt2}
    \right)
    \right\|_2.
\]
From the Lipschitz bound above, for two real inputs
\((\mathbf G_R,\mathbf G_I)\) and
\((\mathbf G_R',\mathbf G_I')\), we have
\[
\begin{aligned}
&
\left|
F(\mathbf G_R,\mathbf G_I)
-
F(\mathbf G_R',\mathbf G_I')
\right|  \\
&\le
\sigma\sqrt{\min(M,K)}
\left\|
\frac{
(\mathbf G_R-\mathbf G_R')
+
i(\mathbf G_I-\mathbf G_I')
}{\sqrt2}
\right\|_F \\
&=
\sigma\sqrt{\frac{\min(M,K)}{2}}
\left(
\|\mathbf G_R-\mathbf G_R'\|_F^2
+
\|\mathbf G_I-\mathbf G_I'\|_F^2
\right)^{1/2} \\
&= \sigma\sqrt{\frac{\min(M,K)}{2}}
\left\| [\mathbf G_R ~ \mathbf G_I] - [\mathbf G_R' ~ \mathbf G_I'] \right\|_F.
\end{aligned}
\]
Therefore, as a function of the underlying real Gaussian vector
\((\mathbf G_R,\mathbf G_I)\), \(F\) is Lipschitz with constant
\begin{equation}
    J=\sigma\sqrt{\frac{\min(M,K)}{2}}.
    \label{eq:Jlip}
\end{equation}
By the standard Gaussian concentration inequality for Lipschitz functions~\cite[Theorem 5.6]{boucheron2013non},
\[
    \mathbb P\left(
    F-\mathbb E F>t
    \right)
    \le
    \exp\left(
    -\frac{t^2}{2J^2}
    \right).
\]
Substituting \(F=\|\mathbf Q_1\|_2\) and using~\eqref{eq:Jlip} gives
\begin{equation}
\mathbb P\left( \|\mathbf Q_1\|_2-\mathbb E\|\mathbf Q_1\|_2>t \right) \le \exp\left( -\frac{t^2}{\sigma^2\min(M,K)} \right). \label{eq:q1conc}
\end{equation}

It remains to bound \(\mathbb E\|\mathbf Q_1\|_2\). To do this, we exploit an  algebraic identity for the mode-1 unfolding $\mathbf{Q}_1 \in \mathbb{C}^{R \times MK}$. By tracing the column-repetition profile of the spatial-smoothing tensor construction, the columns of $\mathbf{Q}_1$ consist of shifted, overlapping blocks of $\mathbf{Q}$. This implies that the row spaces are modulated by a deterministic column-selection operator $\mathbf{H} \in \mathbb{R}^{N \times MK}$ such that $\mathbf{Q}_1 = \mathbf{Q}\mathbf{H}$.

The Gram matrix satisfies:
\[
\mathbf{Q}_1 \mathbf{Q}_1^H = \mathbf{Q} \mathbf{H}\mathbf{H}^H \mathbf{Q}^H = \mathbf{Q} \mathbf{D} \mathbf{Q}^H,
\]
where $\mathbf{D} \in \mathbb{R}^{N \times N}$ is a diagonal matrix whose entries $d_n$ count the number of times the $n$-th column of $\mathbf{Q}$ is repeated across the tensor slices:
\[
d_n = \min(n, \, M, \, K, \, N - n + 1), \quad n = 1, \dots, N.
\]
The maximum entry of $\mathbf{D}$ is bounded by the window dynamics: $d_{\max} = \max_n d_n = \min(M, K)$. This yields:
\begin{align*}
\spec{\mathbf{Q}_1} &= \spec{\mathbf{Q}\mathbf{D}^{1/2}} \le \sigma \sqrt{\min(M, K)} \spec{\mathbf{W}} \\
& = \sigma \sqrt{\min(M, K)} \left(\frac{1}{\sqrt{2}} \spec{\mathbf G_R} + \frac{1}{\sqrt{2}} \spec{\mathbf G_I}\right).
\end{align*}
Invoking the sharp concentration of the spectral norm for standard real Gaussian matrices \cite[Theorem 7.3.1]{vershynin2018high}, we have $\mathbb E\|\mathbf G_R\|_2 \le \sqrt{R} + \sqrt{N}$ and $\mathbb E\|\mathbf G_I\|_2 \le \sqrt{R} + \sqrt{N}$. Therefore, 
\begin{equation}   
\mathbb E\|\mathbf{Q}_1\|_2 \le \sigma \sqrt{2\min(M, K)} \left(\sqrt{R}+\sqrt{N}\right).
\label{eq:expQ1}
\end{equation}
Combining~\eqref{eq:expQ1} with~\eqref{eq:q1conc}, we arrive at the tail bound in Lemma~\ref{lem:mode1noise}.
\end{proof}

We also prove Lemma~\ref{lem:mode2noise}.

\begin{proof} 
Let \(\mathcal{\mathbf H}_2:\mathbb C^{R\times N}\to\mathbb C^{M\times RK}\) 
denote the linear map that sends \(\mathbf Q\) to its mode-2 unfolding \(\mathbf Q_2\). Repeating the arguments from the proof of Lemma~\ref{lem:mode1noise} and using the fact that $\mathcal{\mathbf H}_2(\mathbf E)$ contains the same entries as $\mathcal{\mathbf H}_1(\mathbf E)$ in a different arrangement, we conclude that the mapping \(\mathbf Q\mapsto \|\mathcal{\mathbf H}_2(\mathbf Q)\|_2\) is \(\sqrt{\min(M,K)}\)-Lipschitz. Using Gaussian concentration for Lipschitz functions \cite[Theorem 5.6]{boucheron2013non}, we then have 
\[
\mathbb P\left( \|\mathbf Q_2\|_2-\mathbb E\|\mathbf Q_2\|_2>t \right) \le \exp\left( -\frac{t^2}{\sigma^2\min(M,K)} \right).  
\]

It remains to bound \(\mathbb E\|\mathbf Q_2\|_2\). For \(k=1,\ldots,K\), define 
\[ 
\mathbf Q^{(k)}:=\mathbf Q(:,k:k+M-1)^T\in\mathbb C^{M\times R}. 
\] 
For any \(\mathbf z\in\mathbb C^{RK}\), write \[ \mathbf z= \begin{bmatrix} \mathbf z^{(1)}\\ \vdots\\ \mathbf z^{(K)} \end{bmatrix}, \qquad \mathbf z^{(k)}\in\mathbb C^R. \] 
Then 
\[ 
\mathbf Q_2\mathbf z = \sum_{k=1}^K \mathbf Q^{(k)}\mathbf z^{(k)}. 
\] 
Therefore, 
\[ 
\|\mathbf Q_2\mathbf z\|_2 \le \sum_{k=1}^K \|\mathbf Q^{(k)}\|_2\|\mathbf z^{(k)}\|_2. 
\] 
Since \(\mathbf Q^{(k)}\) is the transpose of a submatrix of \(\mathbf Q\), $\|\mathbf Q^{(k)}\|_2\le \|\mathbf Q\|_2$. Thus 
\[ 
\|\mathbf Q_2\mathbf z\|_2 \le \|\mathbf Q\|_2\sum_{k=1}^K\|\mathbf z^{(k)}\|_2 \le \sqrt K\|\mathbf Q\|_2\|\mathbf z\|_2. 
\] 
Taking the supremum over \(\|\mathbf z\|_2=1\), we get $\|\mathbf Q_2\|_2\le \sqrt K\|\mathbf Q\|_2$. Since \(\mathbf Q=\sigma\mathbf W\), $\mathbb E\|\mathbf Q_2\|_2 \le \sigma\sqrt K\,\mathbb E\|\mathbf W\|_2$. Using the standard expectation bound \cite[Theorem 7.3.1]{vershynin2018high} as in the proof of Lemma~\ref{lem:mode1noise}, we have
$\mathbb E\|\mathbf W\|_2\le \sqrt{2}\left( \sqrt R+\sqrt N \right)$. This gives 
\[ 
\mathbb E\|\mathbf Q_2\|_2 \le \sigma\sqrt{2K}(\sqrt R+\sqrt N). 
\] 
Combining this expectation bound with the concentration inequality proves the result. 
\end{proof}

Next we prove Lemma~\ref{lemma:baseline_alignment_lower}.

\begin{proof}
From the geometric relationship between principal angles and canonical subspace distance metrics, the baseline distance is given by 
\[d(\Umat,\Utrue)=\sqrt{1-\sigmaL^{2}(\Umat^{H}\Utrue)}.\]
Squaring both sides yields the identity 
\begin{equation}
\label{eq:adrelation}
a=\sigmaL^2(\Utrue^H \Umat) = 1 - d^2(\Umat, \Utrue).
\end{equation}

By applying Wedin's $\sin\theta$ theorem~\cite[Ch. V, Sec. 4.1, Theorem 4.4]{StewartSun1990}, \cite[Sec. 2.3]{CaiZhang2018} to the left singular subspaces of $\mathbf{X}_3$ under the deterministic additive noise perturbation $\mathbf{Q}_3$, the subspace distance is bounded by 
\[d(\Umat, \Utrue) \le \frac{\|\mathbf{Q}_3\|_2}{\gamma_3 - \|\mathbf{Q}_3\|_2}.\]
Substituting this upper bound into~\eqref{eq:adrelation} completes the proof.
\end{proof}

Next we prove Lemma~\ref{lemma:mode3_noise_concentration}.

\begin{proof}
Let $\mathbf H_3:\mathbb C^{R\times N}\to\mathbb C^{RM\times K}$ denote the linear map that sends the base noise matrix \(\mathbf Q\) to its mode-3 unfolding \(\mathbf Q_3\). Repeating the arguments from the proof of Lemma~\ref{lem:mode1noise} and using the fact that $\mathcal{\mathbf H}_3(\mathbf E)$ contains the same entries as $\mathcal{\mathbf H}_1(\mathbf E)$ in a different arrangement, we conclude that the mapping \(\mathbf Q\mapsto \|\mathcal{\mathbf H}_3(\mathbf Q)\|_2\) is \(\sqrt{\min(M,K)}\)-Lipschitz. Using Gaussian concentration for Lipschitz functions \cite[Theorem 5.6]{boucheron2013non}, we then have 
\[
\mathbb P\left( \|\mathbf Q_3\|_2-\mathbb E\|\mathbf Q_3\|_2>t \right) \le \exp\left( -\frac{t^2}{\sigma^2\min(M,K)} \right).  
\]

It remains to bound \(\mathbb E\|\mathbf Q_3\|_2\). For any
\(\mathbf z\in\mathbb C^K\), define vectors
\(\mathbf g_m(\mathbf z)\in\mathbb C^N\), \(m=1,\ldots,M\), by
\[
    [\mathbf g_m(\mathbf z)]_{m+k-1}=z_k,
    \qquad k=1,\ldots,K,
\]
and zero elsewhere. Each \(\mathbf g_m(\mathbf z)\) is a shifted copy of
\(\mathbf z\), so
\[
    \|\mathbf g_m(\mathbf z)\|_2=\|\mathbf z\|_2.
\]
By the definition of \(\mathbf Q_3\), the product \(\mathbf Q_3\mathbf z\)
can be written, up to a fixed permutation of entries, as
\[
    \mathbf Q_3\mathbf z
    =
    \begin{bmatrix}
    \mathbf Q\mathbf g_1(\mathbf z)\\
    \mathbf Q\mathbf g_2(\mathbf z)\\
    \vdots\\
    \mathbf Q\mathbf g_M(\mathbf z)
    \end{bmatrix}.
\]
Therefore,
\[
\begin{aligned}
    \|\mathbf Q_3\mathbf z\|_2^2
    &=
    \sum_{m=1}^M
    \|\mathbf Q\mathbf g_m(\mathbf z)\|_2^2 \\
    &\le
    \|\mathbf Q\|_2^2
    \sum_{m=1}^M
    \|\mathbf g_m(\mathbf z)\|_2^2 \\
    &=
    M\|\mathbf Q\|_2^2\|\mathbf z\|_2^2 .
\end{aligned}
\]
Taking the supremum over all \(\|\mathbf z\|_2=1\), we get
\[
    \|\mathbf Q_3\|_2
    \le
    \sqrt M\,\|\mathbf Q\|_2 .
\]
Hence,
\[
    \mathbb E\|\mathbf Q_3\|_2
    \le
    \sqrt M\,\mathbb E\|\mathbf Q\|_2.
\]
Using the standard expectation bound \cite[Theorem 7.3.1]{vershynin2018high} as in the proof of Lemma~\ref{lem:mode2noise}, we have
$\mathbb E\|\mathbf Q\|_2\le \sigma \sqrt{2}\left( \sqrt R+\sqrt N \right)$. This gives 
\[
    \mathbb E\|\mathbf Q_3\|_2
    \le
    \sigma\sqrt{2M}(\sqrt R+\sqrt N).
\]
Combining this expectation bound with the concentration inequality proves
the result.
\end{proof}

Next we prove Lemma~\ref{lem:det}. 

\begin{proof}
The $L$-th SVD relation of $\Xt$ is $\Xt\vLh=\widehat\sigma_L\,\uLh$. Since
$\range(\St)\subseteq\range(\Utrue)$ we have $\UperpH\St=\mathbf 0$, so
\[
\UperpH\Qt\,\vLh=\UperpH\Xt\,\vLh=\widehat\sigma_L\,\UperpH\uLh .
\]
As $\uLh$ is a column of $\Uhat$, $\|\UperpH\uLh\|_2\le\|\UperpH\Uhat\|_2
=d(\Uhat,\Utrue)$; hence $\|\UperpH\Qt\vLh\|_2\le\widehat\sigma_L\,d(\Uhat,\Utrue)$. Weyl's
inequality gives $\widehat\sigma_L=\sigmaL(\Xt)\le\sigmaL(\St)+\|\Qt\|_2=\gamma_3+\|\Qt\|_2$.
Combining yields 
\begin{equation*}
d(\Uhat,\Utrue)\;\ge\;
\frac{\big\|\UperpH\Qt\,\vLh\big\|_2}{\gamma_3+\|\Qt\|_2}.
\end{equation*}
Substituting this lower bound into~\eqref{eq:adrelation} completes the proof.
\end{proof}

Before proving Lemma~\ref{lem:prob}, we establish the following structural result.

Write $\Sig_{\vL}:=\sigma^{-2}\,\E\big[(\Qt\vL)(\Qt\vL)^{\herm}\big]\in\mathbb{C}^{RM\times RM}$ for the
normalized second-moment matrix of the weakest-direction noise leakage.

\begin{lemma}[Structure of the second-moment matrix]\label{lem:secmom}
The second-moment matrix has the Kronecker form $\Sig_{\vL}=\mathbf G\otimes\mathbf I_R$, where $\mathbf G=\mathbf S^{\herm}\mathbf S\in\mathbb C^{M\times M}$ is the Gram matrix of the $M$ unit-norm shifts of $\vL$, collected as the convolution matrix $\mathbf S=[\mathbf g_1\,\cdots\,\mathbf g_M]\in\mathbb C^{N\times M}$
of the length-$K$ filter $\vL$; consequently $\|\Sig_{\vL}\|_2\le\mMK$.
\end{lemma}

\begin{proof}
With $\Qt\vL=[\mathbf Q\mathbf g_1;\dots;\mathbf Q\mathbf g_M]$
(Lemma~\ref{lemma:mode3_noise_concentration} identity) and $\E[\mathbf Q\mathbf B\mathbf Q^{\herm}]=\sigma^2\tr(\mathbf B)\mathbf I_R$ for deterministic
$\mathbf B$, the $(m,m')$ block of $\Sig_{\vL}$ is
$\sigma^{-2}\E[\mathbf Q\mathbf g_m\mathbf g_{m'}^{\herm}\mathbf Q^{\herm}]=(\mathbf g_{m'}^{\herm}\mathbf g_m)\mathbf I_R$, i.e.
\begin{equation}
\Sig_{\vL}=\mathbf G\otimes\mathbf I_R
\label{gir}
\end{equation}
with $\mathbf G=\mathbf S^{\herm}\mathbf S$. Hence $\tr(\Sig_{\vL})=R\,\tr(\mathbf G)=RM$ (each $\|\mathbf g_m\|_2=\|\vL\|_2=1$) 
and $\|\Sig_{\vL}\|_2=\|\mathbf G\|_2$. Two bounds follow: $\|\mathbf G\|_2\le\tr(\mathbf G)=M$ and, since $\mathbf S$ is a finite section of the Toeplitz
operator with symbol
\begin{equation}
\sum_k(v_L)_k e^{-\mathrm i\omega(k-1)},
\end{equation}
$\|\mathbf G\|_2=\|\mathbf S\|_2^2\le\max_\omega\Phi_{\vL}(\omega)\le K$, where
$\Phi_{\vL}(\omega):=\big|\sum_k(v_L)_k e^{-\mathrm i\omega(k-1)}\big|^2$. Together $\|\Sig_{\vL}\|_2\le\mMK$.
\end{proof}

The proof of the lemma above gives the exact Kronecker structure for any unit vector $\mathbf{z} \in \mathbb{C}^K$:
\begin{equation}
\begin{aligned}
\E\big[(\mathbf Q_3\mathbf z)(\mathbf Q_3\mathbf z)^H\big]
&=\sigma^2\,(\mathbf G_{\mathbf z}\otimes\mathbf I_R),
\\
\tr\mathbf G_{\mathbf z}&=M\norm{\mathbf z}_2^2,
\\
\norm{\mathbf G_{\mathbf z}}_2 &\le \min(M,K)\,\norm{\mathbf z}_2^2,
\end{aligned}
\label{eq:lem9}
\end{equation}

We are now in position to prove Lemma~\ref{lem:prob}.

\begin{proof}
We proceed in seven steps.

\medskip\noindent
\emph{Step 1 (Rayleigh quotient).} Define $\mathbf P_{\Utrue} := \Utrue \Utrue^H$. Since $\Uperp\UperpH=\mathbf I_{RM}-\mathbf P_{\Utrue}$
and $\norm{\widehat{\mathbf v}_L}_2=1$,
\begin{align}
\bigl\|\UperpH\mathbf Q_3\widehat{\mathbf v}_L\bigr\|_2^{2}
&=\widehat{\mathbf v}_L^{H}\,\mathbf Q_3^{H}
(\mathbf I_{RM}-\mathbf P_{\Utrue})\mathbf Q_3\,\widehat{\mathbf v}_L
=\widehat{\mathbf v}_L^{H}\mathbf B\,\widehat{\mathbf v}_L
\; \nonumber \\
&\ge\;\lambda_{\min}(\mathbf B)
=\sigma_K\!\bigl((\mathbf I_{RM}-\mathbf P_{\Utrue})\mathbf Q_3\bigr)^2,
\label{eq:rayleigh1}
\end{align}
where
$\mathbf B:=\mathbf Q_3^H(\mathbf I_{RM}-\mathbf P_{\Utrue})\mathbf Q_3
\in\Cplx^{K\times K}$, $\mathbf B\succeq0$. The inequality holds for every unit vector, in particular for the data-dependent
$\widehat{\mathbf v}_L$.

\medskip\noindent
\emph{Step 2 (reduction to a $K\times K$ Gram matrix).}
Let $\mathbf M_0:=\sigma^{-1}(\mathbf I_{RM}-\mathbf P_{\Utrue})\mathbf Q_3
\in\Cplx^{RM\times K}$, so
$\sigma_K((\mathbf I_{RM}-\mathbf P_{\Utrue})\mathbf Q_3)^2
=\sigma^2\lambda_{\min}(\mathbf M_0^H\mathbf M_0)$. Writing
$\bm\Delta:=\mathbf M_0^H\mathbf M_0-\E[\mathbf M_0^H\mathbf M_0]$
and $\lambda_0:=\lambda_{\min}\big(\E[\mathbf M_0^H\mathbf M_0]\big)$,
Weyl's inequality gives
\begin{equation}\label{eq:weyl}
\lambda_{\min}(\mathbf M_0^H\mathbf M_0)
\;\ge\;\lambda_0-\norm{\bm\Delta}_2 .
\end{equation}

\medskip\noindent
\emph{Step 3 (deterministic floor for  $\lambda_0$).}
Let $m := \min(M,K)$. Since $\mathbf I_{RM}-\mathbf P_{\Utrue}$ is an orthogonal projector,
$\mathbf M_0^H\mathbf M_0
=\sigma^{-2}\mathbf Q_3^H\mathbf Q_3
-\sigma^{-2}\mathbf Q_3^H\mathbf P_{\Utrue}\mathbf Q_3$, hence
\[
\E[\mathbf M_0^H\mathbf M_0]=RM\,\mathbf I_K-\mathbf D,
\]
where $\mathbf D:=\sigma^{-2}\,\E\big[\mathbf Q_3^H\mathbf P_{\Utrue}\mathbf Q_3\big]
\succeq0$. Here $\E[\mathbf Q_3^H\mathbf Q_3]=RM\sigma^2\,\mathbf I_K$ because
distinct time indices of $\mathbf W$ are independent:
$(\E[\mathbf Q_3^H\mathbf Q_3])_{jk}
=\sigma^2\sum_{r,m}\E[W_{r,j+m-1}^*\,W_{r,k+m-1}]
=RM\sigma^2\delta_{jk}$. For any unit $\mathbf z\in\Cplx^K$, using
\eqref{eq:lem9}, and cyclicity of the trace,
\begin{align*}
\mathbf z^H\mathbf D\mathbf z
&=\sigma^{-2}\,\E\norm{\mathbf P_{\Utrue}\mathbf Q_3\mathbf z}_2^2
=\tr\!\big(\mathbf P_{\Utrue}(\mathbf G_{\mathbf z}\otimes\mathbf I_R)\big) \\
&\le\rank(\mathbf P_{\Utrue})\,
\norm{\mathbf G_{\mathbf z}\otimes\mathbf I_R}_2
=L\norm{\mathbf G_{\mathbf z}}_2\le Lm ,
\end{align*}
so $\lambda_{\max}(\mathbf D)\le Lm$ and
$\lambda_0=RM-\lambda_{\max}(\mathbf D)\ge RM-Lm$.

\medskip\noindent
\emph{Step 4 (tail bound with explicit constant).}
Fix a unit $\mathbf z\in\Cplx^K$ and set
$D_{\mathbf z}
:=\norm{\mathbf M_0\mathbf z}_2^2-\E\norm{\mathbf M_0\mathbf z}_2^2$.
The vector $\mathbf y:=\mathbf M_0\mathbf z$ is a zero-mean Gaussian random vector with covariance
\[
\mathbf{Cov}_{\mathbf y}
:=\E[\mathbf y\mathbf y^H]
=(\mathbf I_{RM}-\mathbf P_{\Utrue})
(\mathbf G_{\mathbf z}\otimes\mathbf I_R)
(\mathbf I_{RM}-\mathbf P_{\Utrue}),
\]
by \eqref{eq:lem9}. Diagonalize
$\mathbf{Cov}_{\mathbf y}=\mathbf V\bm\Lambda\mathbf V^H$ with
eigenvalues $\lambda_1\ge\dots\ge0$; then
$\mathbf y{=}\mathbf V\bm\Lambda^{1/2}\mathbf w$ with
$\mathbf w$ standard complex Gaussian, so
\[
\norm{\mathbf y}_2^2
{=}\sum_i\lambda_i\,\lvert w_i\rvert^2
=\sum_i\frac{\lambda_i}{2}\,(\xi_i^2+\eta_i^2),
\qquad \xi_i,\eta_i\ \text{i.i.d.\ }N(0,1).
\]
Thus $\norm{\mathbf M_0\mathbf z}_2^2$ is a real Gaussian quadratic
form $\mathbf g^T\mathbf A_{\mathbf z}\mathbf g$,
$\mathbf g\sim N(0,\mathbf I)$, whose matrix
$\mathbf A_{\mathbf z}$ has nonzero spectrum
$\{\lambda_i/2,\ \text{each with multiplicity }2\}$; consequently $\tr\mathbf A_{\mathbf z}=\tr\mathbf{Cov}_{\mathbf y}$, $\norm{\mathbf A_{\mathbf z}}_F^2
=\tfrac12\tr(\mathbf{Cov}_{\mathbf y}^2)$, and $\norm{\mathbf A_{\mathbf z}}_2
=\tfrac12\norm{\mathbf{Cov}_{\mathbf y}}_2$.

The following dimensional bounds follow from \eqref{eq:lem9} and the fact that
projections contract PSD traces and spectral norms:
$\tr\mathbf{Cov}_{\mathbf y}\le R\,\tr\mathbf G_{\mathbf z}=RM$,
$\norm{\mathbf{Cov}_{\mathbf y}}_2
\le\norm{\mathbf G_{\mathbf z}}_2\le m$, and
$\tr(\mathbf{Cov}_{\mathbf y}^2)
\le\norm{\mathbf{Cov}_{\mathbf y}}_2\tr\mathbf{Cov}_{\mathbf y}
\le RMm$; hence
$\norm{\mathbf A_{\mathbf z}}_F^2\le RMm/2$ and
$\norm{\mathbf A_{\mathbf z}}_2\le m/2$. The Gaussian Hanson--Wright
inequality with explicit constant $\tfrac18$
\cite[Lemma.~4]{lamperski2023nonasymptotic}, applied to $\mathbf A_{\mathbf z}$ and to
$-\mathbf A_{\mathbf z}$, then yields the following guarantee for all $r > 0$:
\begin{align}
\Prob(\lvert D_{\mathbf z}\rvert>r)
&\le2\exp\!\Big[-\tfrac18\min\Big(\tfrac{r^2}{RMm/2},
\tfrac{r}{m/2}\Big)\Big] \nonumber 
\\ &=2\exp\!\Big[-\tfrac14\min\Big(\tfrac{r^2}{RMm},
\tfrac{r}{m}\Big)\Big]. \label{eq:phanson}
\end{align}

\medskip\noindent
\emph{Step 5 ($\varepsilon$-net for the operator
norm).}
$\bm\Delta$ is Hermitian and
$\mathbf z^H\bm\Delta\mathbf z=D_{\mathbf z}$, so
$\norm{\bm\Delta}_2=\sup_{\norm{\mathbf z}_2=1}
\lvert D_{\mathbf z}\rvert$. Let $\mathcal N$ be an $\varepsilon$-net of
the unit sphere of $\Cplx^K$ (real dimension $2K$), so
$\lvert\mathcal N\rvert
\le(1+2/\varepsilon)^{2K}$
\cite[Corollary.~4.2.11]{vershynin2018high}. For any unit $\mathbf z$ pick
$\widehat{\mathbf z}\in\mathcal N$ with
$\norm{\mathbf z-\widehat{\mathbf z}}_2\le\varepsilon$;
then
$\lvert D_{\mathbf z}-D_{\widehat{\mathbf z}}\rvert
=\lvert\mathbf z^H\bm\Delta(\mathbf z-\widehat{\mathbf z})
+(\mathbf z-\widehat{\mathbf z})^H\bm\Delta\widehat{\mathbf z}\rvert
\le2\varepsilon\norm{\bm\Delta}_2$, whence
$\norm{\bm\Delta}_2\le(1-2\varepsilon)^{-1}
\max_{\mathcal N}\lvert D_{\widehat{\mathbf z}}\rvert$. A union bound
over $\mathcal N$ with~\eqref{eq:phanson} at level
$(1-2\varepsilon)r$ gives, for all $r>0$,
\begin{align}
&\Prob\big(\norm{\bm\Delta}_2>r\big) \nonumber \\
&\le\Big(1+\tfrac{2}{\varepsilon}\Big)^{2K}\cdot
2\exp\!\Big[-\tfrac14\min\Big(
\tfrac{\big((1-2\varepsilon)r\big)^2}{RMm},\;
\tfrac{(1-2\varepsilon)r}{m}\Big)\Big].
\label{eq:hansen2}
\end{align}

\medskip\noindent
\emph{Step 6 (specialization).}
Set $r = \mu RM$. The quadratic branch of the minimum in~\eqref{eq:hansen2} is the smaller one iff
$(1-2\varepsilon)r\le RM$, i.e.\
$\mu(1-2\varepsilon)\le1$, which holds automatically for $\mu<1$. Hence~\eqref{eq:hansen2} collapses to the closed form
\[
\Prob\big(\norm{\bm\Delta}_2>\mu RM\big)
\le p_{\mathrm{w}}(\varepsilon,\mu).
\]

\medskip\noindent
\emph{Step 7 (assembly).}
On the event $\{\norm{\bm\Delta}_2\le\mu RM\}$, Steps 2--3 give
$\lambda_{\min}(\mathbf M_0^H\mathbf M_0)\ge RM-Lm-\mu RM = RM(1-\mu) - Lm$. If
the right-hand side is negative, the claim \eqref{star_mu} is trivially
true; otherwise take square roots and multiply by $\sigma$. Combining
with~\eqref{eq:rayleigh1} completes the proof of Lemma~\ref{lem:prob}.
\end{proof}

Next we prove Lemma~\ref{lem:incage}.

\begin{proof}
The thin SVD relation $\Xt\Vhat=\Uhat\Shat$ and $\UperpH\St=\mathbf{0}$ give $\UperpH\Qt\Vhat=\UperpH\Uhat\Shat$, so
$\UperpH\Uhat=\UperpH\Qt\Vhat\,\Shat^{-1}$. Hence
\begin{align}
(\I-\PiS)\Uhat &=\Uperp\UperpH\Uhat \nonumber \\
&=(\I-\PiS)\Qt\Vhat\Shat^{-1}. \label{eq:proj1}
\end{align}
Because $\mathrm{range}(\PiS)\subseteq\mathrm{range}(\PK)$ we have $\PK\PiS=\PiS$, so
$(\PK-\PiS)(\I-\PiS)=\PK-\PiS$; left-multiplying~\eqref{eq:proj1} by $(\PK-\PiS)$ then yields 
\begin{equation*}
\Epar=(\PK-\PiS)\,\Qt\,\Vhat\,\Shat^{-1},
\label{eq:identity}
\end{equation*}
Using $\twonorm{\Shat^{-1}}=1/\sigmaL(\Xt)$ and submultiplicativity, we obtain
\[
\twonorm{\Epar}\;\le\;\frac{\twonorm{(\PK-\PiS)\,\Qt\,\Vhat}}{\sigmaL(\Xt)}
\]
Using $\twonorm{\Vhat}=1$ and Weyl with $\sigmaL(\Xt)\ge\gamma_3-\twonorm{\Qt}$, we obtain~\eqref{eq:incage-bound}.
\end{proof}

Next we prove Lemma~\ref{lem:incageenv}.

\begin{proof}
Since $\E[\Qt\Qt^{H}]=\sigma^2 K\,\I_{RM}$ and $\operatorname{tr}(\PK-\PiS)=L^2-L$, we have
$\E\,\fnorm{(\PK-\PiS)\Qt}^2=\sigma^2 K(L^2-L)$, hence
$\E\twonorm{(\PK-\PiS)\Qt}\le\sigma\sqrt{(L^2-L)K}$. The lemma follows from the fact that the mapping \(\mathbf Q\mapsto \|(\PK-\PiS)\mathcal{\mathbf H}_3(\mathbf Q)\|_2\) is \(\sqrt{\min(M,K)}\)-Lipschitz.
\end{proof}

\subsection{Supporting proofs for Section~\ref{sec:finalproof}}

We first characterize the topology of $\gorlower$ to justify the substitution of our directional inequalities.

\begin{lemma}[Monotonicity]
\label{lem:monotonicity}
Let $g(a,E):=\sqrt{1-a}-\sqrt{E/(a+E)}$ on $a\in(0,1),\,E\ge0$. Let
\[
a^* := \frac{4\sqrt{7}-2}{27} \approx 0.318.
\]
If $a \ge a^*$, then $g(a, E)$ is non-increasing in both variables $a$ and $E$. Consequently, if $a^* \le a \le \overline{a} <1$ and $0 \le E \le \overline{E}$, then $g(a,E) \ge g(\overline{a}, \overline{E})$.
\end{lemma}
\begin{proof}
Note that 
\begin{align*}
\frac{\partial g}{\partial E}=-\frac{a}{2\sqrt{E}\,(a+E)^{3/2}}<0
\end{align*}
since $a,E>0$. Therefore, $g$ is non-increasing in $E$.

In addition, 
\[
\frac{\partial g}{\partial a}=-\frac{1}{2\sqrt{1-a}}+\frac{\sqrt{E}}{2(a+E)^{3/2}}.
\]
Hence, $\frac{\partial g}{\partial a} \le 0$ if and only if $E(1-a) \le (a+E)^3$.

For fixed $a$, define
\[
h_a(E) := \frac{E (1-a)}{(a+E)^3}. 
\]
The function $h_a(E)$ is maximized over $E \ge 0$ at $E = \frac{a}{2}$. The maximum value of $h_a(E)$ at that point is 
\[
h_a\Big(\frac{a}{2}\Big) = \frac{4(1-a)}{27a^2.}
\]
To ensure $h_a(a/2) \le 1$ for any $E \ge 0$, we require $a \ge a^*$. 

Therefore, we conclude that $g$ is non-increasing in both variables on $[a^*, 1) \times [0, \infty)$. 
\end{proof}

We are now ready to prove Theorem~\ref{thm:main}.

\begin{proof}
{\em Statement 1 (oracle gain):} Lemma~\ref{lemma:mode3_noise_concentration} gives $\|\mathbf{Q}_3\|_2 \le \omega_3$ with failure probability at most $e^{-t_3^2/(\sigma^2\mMK)}$. On this event, Lemma~\ref{lemma:baseline_alignment_lower} implies $a \ge \underline{a}$. On this same event, Lemma~\ref{lem:det} and Lemma~\ref{lem:prob} give $a \le \overline{a}$ with failure probability at most $p_{\mathrm{w}}(\varepsilon,\mu)$. On this same event, Lemma~\ref{lem:incage} and Lemma~\ref{lem:incageenv} give $\spec{\Epar}^2 \le \overline{E}$ with failure probability at most $e^{-t_4^2/(\sigma^2\mMK)}$. Taking a union bound gives a failure probability bounded by the sum of the failure probabilities above. Lemma~\ref{lem:monotonicity} then gives $g(a,\spec{\Epar}^2) \ge g(\overline{a}, \overline{E})$, i.e., $\gorlower \ge \underline{\gorlower}$. The conclusion~\eqref{eq:finalgor} then follows from Theorem~\ref{thm:oraclebounds}.

{\em Statement 2 (empirical loss):} 
Lemma~\ref{lem:mode1noise} gives $\|\mathbf{Q}_1\|_2 \le \omega_1$ with failure probability at most $e^{-t_1^2/(\sigma^2\mMK)}$, Lemma~\ref{lem:mode2noise} gives $\|\mathbf{Q}_2\|_2 \le \omega_2$ with failure probability at most $e^{-t_2^2/(\sigma^2\mMK)}$, and Lemma~\ref{lemma:mode3_noise_concentration} gives $\|\mathbf{Q}_3\|_2 \le \omega_3$ with failure probability at most $e^{-t_3^2/(\sigma^2\mMK)}$. Taking a union bound gives a failure probability bounded by the sum of the failure probabilities above. Combining Lemma~\ref{lem:rholower} and Lemma~\ref{lemma:baseline_alignment_lower} then gives $\rho \ge \underline{\rho} = \sqrt{\underline{a}}$, and Corollary~\ref{cor:wedin_sin} gives $\eta \le \overline{\eta}$. Since $\overline{\eta} \le \frac{1}{2} \underline{\rho}$, we also have $\eta < \rho$. It follows that $\frac{\eta}{\rho-\eta} \le \frac{\overline{\eta}}{\underline{\rho}-\overline{\eta}}$, i.e., $\lemupper \le \overline{\lemupper}$. The conclusion~\eqref{eq:finallem} then follows from Theorem~\ref{thm:exact}.

{\em Statement 3 (overall tensor gain):} Taking a union bound gives a failure probability bounded by the sum of the failure probabilities above; the failure probability of the event $\|\mathbf{Q}_3\|_2 \le \omega_3$ need be counted only once. The conclusion~\eqref{eq:finalgain} follows directly from~\eqref{eq:overalltg2} together with~\eqref{eq:finalgor} and~\eqref{eq:finallem}.
\end{proof}

\end{document}